\documentclass[3p,nopreprintline]{elsarticle}

\usepackage{hyperref}

\usepackage{amssymb}
\usepackage{appendix}
\usepackage{mathtools}
\usepackage[caption=false]{subfig} 
\usepackage{enumitem} 
\usepackage{microtype}
\usepackage{algorithm}
\usepackage{algorithmic}
\usepackage{tikz}
\usepackage{sts}
\usepackage{thm-restate} 

\makeatletter
\renewcommand{\ALG@name}{Procedure}
\makeatother
\newtheorem{proposition}{Proposition}
\newtheorem{lemma}[proposition]{Lemma}
\newtheorem{corollary}[proposition]{Corollary}
\newtheorem{claim}[proposition]{Claim}
\newdefinition{definition}{Definition}
\newdefinition{example}{Example}
\newdefinition{remark}{Remark}
\newproof{proof}{Proof}

\begin{document}

	\begin{frontmatter}
	\title{Decisiveness of Stochastic Systems and its Application to Hybrid Models\tnoteref{tnote}}
	\tnotetext[tnote]{Research partly supported by F.R.S.-FNRS under Grant n$^\circ$ F.4520.18 (ManySynth) and Grant n$^\circ$ T.0027.21, by ENS Paris-Saclay visiting professorship (M.\ Randour, 2019), and by ANR MAVeriQ (ANR-20-CE25-0012).}
	\author[1]{Patricia Bouyer}
	\author[2]{Thomas Brihaye}
	\author[2,3]{Mickael Randour\texorpdfstring{\fnref{fnmick}}{}}
	\author[2]{C\'edric Rivi\`ere}
	\author[1,2,3]{Pierre Vandenhove\texorpdfstring{\fnref{fnpierre}}{}}

	\address[1]{Université Paris-Saclay, CNRS, ENS Paris-Saclay, Laboratoire Méthodes Formelles, 91190, Gif-sur-Yvette, France}
	\address[2]{UMONS -- Universit{\'e} de Mons, Mons, Belgium}
	\address[3]{F.R.S.-FNRS, Belgium}
	\fntext[fnmick]{Mickael Randour is an F.R.S.-FNRS Research Associate.}
	\fntext[fnpierre]{Pierre Vandenhove is an F.R.S.-FNRS Research Fellow.}

	\begin{abstract}
		In~\cite{ABM07}, Abdulla et al.\ introduced the concept of
		decisiveness, an interesting tool for lifting good properties of
		finite Markov chains to denumerable ones. Later~\cite{BBBC18}, this
		concept was extended to more general stochastic transition systems (STSs), allowing the
		design of various verification algorithms for large classes of (infinite)
		STSs. We further improve the understanding and utility of decisiveness in two ways.

		First, we provide a general criterion for proving the decisiveness of
		general STSs. This criterion, which is very
		natural but whose proof is rather technical, (strictly) generalizes
		all known criteria from the literature.

		Second, we focus on stochastic hybrid systems (SHSs), a
		stochastic extension of hybrid systems. We establish the
		decisiveness of a large class of SHSs and, under a few classical
		hypotheses from mathematical logic, we show how to decide reachability
		problems in this class, even though they are undecidable for
		general SHSs. This provides a decidable stochastic extension of
		o-minimal hybrid systems~\cite{BMRT04,Gen05,LPS00}.
	\end{abstract}
	\end{frontmatter}

	\section{Introduction}
		\paragraph{Hybrid and stochastic models}
Various kinds of mathematical models have been proposed to represent real-life
systems. In this article, we focus on models combining \emph{hybrid} and
\emph{stochastic} aspects. We outline the main features of these models to
motivate our approach.

The idea of \emph{hybrid systems} originates from
the urge to study systems subject to both \emph{discrete} and \emph{continuous}
phenomena, such as digital computer systems interacting with analog data.
These systems are transition systems with two kinds of transitions: continuous
transitions, where some continuous variables evolve over time (e.g., according
to a differential equation), and discrete transitions, where the system changes
modes and variables can be reset.
Much of the research about hybrid systems focuses on \emph{non-deterministic}
hybrid systems, i.e., systems modeling uncertainty by considering all possible
behaviors (e.g., different possibilities for discrete transitions
at a given time, arbitrary long time between discrete transitions). A typical
question concerns the \emph{safety} of such systems---if a system can
reach an undesirable state, it is said to be \emph{unsafe}; if not, it is
\emph{safe}. If there is
even a single behavior that does not satisfy the properties that the system
should verify, then the whole system is deemed inadequate---such a
specification is called \textit{qualitative}. However, this is limiting for two
reasons. First, it does not take
into account that some behaviors are more likely to occur than others. Second,
risks cannot necessarily be avoided, and it is unrealistic to prevent
undesirable outcomes altogether. Therefore, we want to make probabilities an
integral part of our models,
in order to be able to \textit{quantify} the probability that
they behave according to the specification.
We thus consider the class of \emph{stochastic hybrid systems} (SHSs, for
short), hybrid systems in which a stochastic semantics replaces non-determinism.

\paragraph{Goals}
Our interest lies in the formal analysis of continuous-time SHSs, and more
specifically in \emph{reachability} questions, i.e., concerning the likelihood
that some set of states is reached in a system. The
questions we seek to
answer are both of the qualitative kind (is some region of the state space
\emph{almost surely} reached, i.e., reached with probability $1$?) and of
the quantitative kind (what is the probability
that some part of the state space is eventually reached?). Such questions are
crucial, as verifying that a system works safely often reduces to verifying
that some undesirable state of the system is never reached (or reached with a
very low probability), or that some desirable state is to be reached with high
probability~\cite{BK08}.
We want to give algorithms that decide, for an SHS $\calH$ and
reachability property $P$, whether $P$ is satisfied in $\calH$.
Such an endeavor faces multiple challenges; a first obvious one being that even
for rather restricted classes of non-deterministic hybrid systems, reachability
problems are undecidable~\cite{HKPV,HR00}. We want to define and
consider a \emph{class} of SHSs for which \emph{some} reachability problems are
decidable.

\paragraph{Methods and contributions}
Our methodology consists of two main steps.  In a first step, we
follow the approach of Bertrand et al.~\cite{BBBC18}: we study general
\emph{stochastic transition systems} (STSs) through the
\emph{decisiveness} concept (Section~\ref{sec:sts}). The class of STSs
is a very versatile class of systems encompassing many well-known
families of stochastic systems, such as Markov chains, but also
stochastic systems with a continuous state-space such as generalized
semi-Markov processes, stochastic timed automata, stochastic Petri
nets, and stochastic hybrid systems.  Decisiveness was introduced
in~\cite{ABM07} to study Markov chains, and extended to STSs
in~\cite{BBBC18}. An STS is said to be \emph{decisive} with respect to
a set of states $B$ if executions of the system almost surely reach
either $B$ or a state from which $B$ is unreachable. Decisive
STSs benefit from many useful properties that make possible the design of
some verification algorithms related to reachability properties. Our
first contribution is to provide a \textbf{criterion to check the
  decisiveness of STSs} (Proposition~\ref{prop:mainDecCritGen}), which
generalizes the decisiveness criteria from~\cite{ABM07,BBBC18}. This
generalization was mentioned as an open problem in~\cite{BBBC18}.

In a second step, we focus on \textit{stochastic hybrid systems}, which we
introduce in Section~\ref{sec:stochHS}. Our contributions regarding SHSs are
split into three parts.

First, we motivate our subsequent results by showing a proof that \textbf{reachability problems are undecidable in general for SHSs} (Proposition~\ref{prop:undecSHS}).
More precisely, our proof shows the undecidability of a subclass of SHSs which is different from those appearing in classical proofs from the literature for non-deterministic hybrid systems~\cite{HKPV,HR00}.
Our intent with this new undecidability proof is to get as close as possible to the decidable class that we consider in the following sections.

Second, we aim to use the decisiveness idea to get closer to
the decidability frontier.
Albeit desirable, the decisiveness of a class of SHSs is not
sufficient to handle algorithmic questions about each SHS, as we need an
effective way to apprehend their uncountable state space. In this regard, an
often-used technique is to consider
a \emph{finite abstraction} of the system, that is, a finite partition of the
state space that preserves the properties to be verified (a well-known
example is the region graph for timed automata~\cite{AD94}).
To find such an abstraction of SHSs, we borrow ideas from~\cite{BMRT04,LPS00}:
we consider SHSs with \emph{strong resets} (Section~\ref{sec:decisiveSHS}), a
syntactic condition that decouples their continuous behavior from
their discrete behavior. We show that \textbf{SHSs with adequately placed
strong resets} (at least one per cycle of their discrete graph) $(i)$
\textbf{have a finite abstraction} (Proposition~\ref{prop:finiteAbs}), and
$(ii)$ \textbf{are decisive} (Proposition~\ref{prop:strongResetAreDec}), which
can be proved using our new criterion.

Third, in Section~\ref{sec:ominimal}, we show, under the hypotheses of
the previous section, how to effectively compute a finite abstraction and use
it to perform a reachability analysis of the original system.
The way we proceed is by assuming that the components of our systems are
definable in an \emph{o-minimal structure}. The main difficulty here lies in
the fact that ``\emph{a satisfactory theory of measure and integration seems to
be lacking in the o-minimal context}''~\cite{BO04}. In particular, in an
o-minimal structure, the primitive function of a definable function
is in general not definable in the same o-minimal structure, which complicates
definability questions regarding probabilities. We therefore slightly restrict
the possible probability distributions that can be used, using properties of
o-minimal structures to keep our class as large as possible.
When the \textbf{theory} of the structure is \textbf{decidable} (as is the case
for the ordered field of real numbers~\cite{Tar48}), the \textbf{reachability
problems then become decidable}. This provides a stochastic extension to the
theory of o-minimal hybrid systems~\cite{LPS00}.
We study both \textbf{qualitative} (Section~\ref{sec:qualReachAnal}) and
\textbf{quantitative} (Section~\ref{sec:quantReachDec}) problems. Some proofs and technical details are deferred to the appendix.

\paragraph{Related work}
Our results combine previous work on stochastic systems and hybrid
systems. About stochastic systems, we build on works about the decisiveness notion~\cite{ABM07,BBBC18}.
Our work also takes inspiration from research about \emph{stochastic timed
automata}, a subclass of stochastic hybrid systems which already combines
stochastic and timed aspects;
model checking of stochastic timed automata has been studied
in~\cite{BBBBG08,BBBM08,BBB+14} and
considered in the context of the decisiveness property in~\cite{Car17}.
Our model of stochastic hybrid systems can be seen as a generalization of this model of stochastic timed automata.
Fundamental results about the decidability of the reachability problem for
simple classes of (non-deterministic) hybrid systems can be found in~\cite{AD94,HKPV,HR00}.
The class of \emph{o-minimal hybrid systems}, of which we
introduce a stochastic extension, has been studied
in~\cite{BMRT04,Gen05,LPS00}.

Our aim about SHSs is mostly theoretical: we show a path to obtain a large decidable class, using the decisiveness notion and decidability results from mathematical logic as our main tools.
The literature about SHSs often follows a more practical or numerical
approach. A first introduction to the model was provided in~\cite{HLS00}. An
extensive review of the underlying theory and of many applications of SHSs
is provided in book~\cite{CL07}, and a review of different possible
semantics for this model is provided in~\cite{LP10}.
Book~\cite{Bu12} overviews several approaches (probabilistic, analytic, statistical) to the reachability problem in general stochastic hybrid systems.
Our point of view is orthogonal to these existing approaches, yielding interesting formal guarantees at the expense of syntactic restrictions, which we discuss in Sections~\ref{sec:stochHS},~\ref{sec:decisiveSHS}, and~\ref{sec:ominimal}.
Applications of SHSs are numerous: a few examples are air traffic
management~\cite{PH09,PHLS00}, communication networks~\cite{He04}, biochemical
processes~\cite{LOQD17,SH10}. The software tool \textsc{Uppaal}
implements a model of SHS similar to the one studied in this article~\cite{DDL+12} and uses
numerical methods to compute reachability probabilities, through numerical
solving of differential equations and Monte Carlo simulation.
Reachability problems have also been considered in an alternative semantics
with \emph{discrete time}; a numerical approach is for instance provided
in~\cite{AKLP10,APLS08}.

This article extends its conference version~\cite{BBRRV20}.

\paragraph{Notations}
We write $\IN = \{0, 1, 2,\ldots\}$ for the set of non-negative integers, and $\IR^+ = \{x\in\IR\mid x \ge 0\}$ for the set of
non-negative real numbers. To emphasize that a union is
disjoint, we denote it by $\disjUnion$ instead of $\bigcup$.
Let $(\Omega,\Sigma)$ be a measurable space. We write
$\Dist(\Omega,\Sigma)$ (or $\Dist(\Omega)$ if there is no ambiguity) for
the set of \emph{probability} distributions over $(\Omega,\Sigma)$.
The complement of a set $A\in\Sigma$ is denoted by $\comp{A} =
\Omega\setminus A$.
For $A\in\Sigma$ a measurable set,
we say that two probability distributions $\mu, \nu \in \Dist(\Omega)$ are
\emph{qualitatively equivalent on $A$} (or \emph{equivalent on $A$}) if for
each $B\in\Sigma$, if $B\subseteq A$, then $\mu(B)>0$ if and only if $\nu(B)>0$.
For $s\in \Omega$, we denote by $\delta_s$ the Dirac distribution centered on
$s$.

	\section{Decisiveness of Stochastic Transition Systems}
		\label{sec:sts}
		In this section, we define \emph{stochastic transition systems} (STSs,
for short) as in~\cite{BBBC18}.
We then describe
the concept of \emph{decisiveness}, first defined in the specific case
of Markov chains~\cite{ABM07}, and then extended to STSs~\cite{BBBC18}.
Decisive stochastic systems benefit from ``nice'' properties making their
qualitative and quantitative analysis more accessible. A first
contribution of our work is a new decisiveness criterion
(Proposition~\ref{prop:mainDecCritGen}), which generalizes existing criteria
from the literature~\cite{ABM07,BBBC18}. It is an intuitive criterion,
which was conjectured in~\cite{BBBC18} but could not be proved.
We finish with a brief subsection on the notion of \emph{abstraction}
between STSs, which will be useful to apply our results to stochastic
hybrid systems.
We will use STSs in Section~\ref{sec:stochHS} to define the semantics of
\emph{stochastic hybrid systems}, a model of STSs to which the newly
developed techniques will apply.

\subsection{Stochastic Transition Systems\texorpdfstring{~\cite[Section~2]{BBBC18}}{}}
\begin{definition}[Stochastic transition system]
	A \emph{stochastic transition system} (STS) is a tuple
	$\calT =(S,\Sigma,\kappa)$ consisting of a measurable space of \emph{states}
	$(S,\Sigma)$, and a function $\kappa\colon S \times\Sigma \to [0,1]$ (called \emph{Markov kernel}) such that
	for every $s \in S$, $\kappa(s,\cdot)$ is a probability measure
	and for every $A \in \Sigma$, $\kappa(\cdot, A)$ is a measurable function.
\end{definition}
The second condition on $\kappa$ implies in particular that for a
measurable set $B\in\Sigma$, the set
\[\Pre^\calT(B) = \{s\in S\mid \kappa(s, B) > 0\}\]
is measurable.

\paragraph{Measuring runs}
We interpret STSs as systems generating executions, with a probability
measure over these executions.
We fix an STS $\calT=(S,\Sigma,\kappa)$. From a state $s \in S$, a
probabilistic transition is performed according to distribution
$\kappa(s, \cdot)$, and the system resumes from one of the successor states;
this process generates random sequences of states.
To formally provide a probabilistic semantics to STSs, we define a probability
measure over the set of \emph{runs} of STSs.
A \emph{run} of $\calT$ is an infinite
sequence $\run = s_0s_1s_2\ldots$ of states. We write $\Runs(\calT)$
for the set of runs of $\calT$.

In order to get a probability measure over $\Runs(\calT)$, we equip
this set with a $\sigma$-algebra. For every finite sequence
$(A_i)_{0 \le i \le n} \in \Sigma^{n+1}$, we define the
\emph{cylinder}
\[
\Cyl(A_0,A_1,\ldots,A_n) = \{\run = s_0 s_1 \ldots s_n\ldots \in
\Runs(\calT) \mid \forall 0 \le i \le n, s_i \in A_i\}\puncteq{.}
\]
Given an initial distribution $\mu\in\Dist(S)$, we extend in a natural way the Markov kernel $\kappa$ to a function $\Prob_{\mu}^\calT$ on the set of all cylinders.
We initialize it for $A_0 \in \Sigma$ with
$\Prob_{\mu}^{\calT}(\Cyl(A_0))=\mu(A_0)$: this is the probability to
be in $A_0$, starting with probability distribution $\mu$.
For every finite sequence $(A_i)_{0 \le i \le n} \in \Sigma^{n+1}$, we
then set inductively:
\[
\Prob^{\calT}_\mu(\Cyl(A_0,A_1,\ldots,A_n)) = \int_{s_0 \in A_0}
\Prob^{\calT}_{\kappa(s_0,\cdot)}(\Cyl(A_1,\ldots,A_n))\ud\mu(s_0)
\puncteq{.}
\]
The intuition is that we split the probability on the left-hand side into all
the possible ways to perform the first step in $A_0$, according to the initial
distribution $\mu$.
Function $\Prob_{\mu}^\calT$ can then be lifted to a unique probability measure on the $\sigma$-algebra generated by all the cylinders, using Ionescu-Tulcea extension theorem~\cite[Section~2.7.2]{Ash72}.

\paragraph{Expressing properties of runs}
To express properties of runs of $\calT$, we use standard notations taken from \LTL~\cite{Pnu77}.
If $B, B' \in \Sigma$, we write in particular $\F B$ (resp.\ $\F[\le n] B$, $B'\U B$, $B'\U[\le n] B$, $\GG B$, $\GG \F B$) for the set of runs that visit $B$ at some point (resp.\ visit $B$ in less than $n$ steps, stay in $B'$ until a first visit to $B$, stay in $B'$ until a first visit to $B$ in less than $n$ steps, always stay in $B$, visit $B$ infinitely often).
We postpone more formal definitions to~\ref{sec:app-defs}.
We will be especially interested in two kinds of reachability problems.

\begin{definition}[Qualitative and quantitative reachability]
\label{def:problems}
	Let $B\in\Sigma$ be a measurable set of target states, and $\mu\in\Dist(S)$
	be an initial distribution. The \emph{qualitative reachability problems}
	consist in deciding whether $\Prob^\calT_\mu(\F B) = 1$, and whether
	$\Prob^\calT_\mu(\F B) = 0$. The
	\emph{quantitative reachability problem} consists in deciding, given
	$\epsilon, p\in\IQ$ with $\epsilon > 0$, whether $\abs{\Prob^\calT_\mu(\F
	B) - p} < \epsilon$.
\end{definition}

\paragraph{Transforming probability distributions}
Another useful way to reflect on STSs is as transformers of
probability distributions on $(S,\Sigma)$.
\begin{definition}[STS as a transformer]
	For $\mu \in\Dist(S)$, its \emph{transformation through $\calT$} is the
	probability distribution $\Omega_\calT(\mu)\in\Dist(S)$ defined for
	$A\in\Sigma$ by
	\[
	\Omega_\calT(\mu)(A) = \int_{s\in S} \kappa(s, A)\, \ud\mu(s)
	= \Prob^\calT_\mu(\Cyl(S,A))
	\puncteq{.}
	\]
\end{definition}
The meaning of function $\Omega_\calT$ can be interpreted as
follows: $\Omega_\calT(\mu)(A)$ is the probability to reach $A$
in one step, from the initial distribution $\mu$.

\begin{definition}[Conditional distributions] \label{def:condDist}
       For $\mu \in\Dist(S)$ and $A \in \Sigma$ such that
       $\mu(A) > 0$, the \emph{conditional probability distribution of
       $\mu$ given $A$} is denoted as $\mu_A$ and is such that for
       $B\in\Sigma$, $\mu_A(B) = \mu(A\cap B)/\mu(A)$.

	Let $\calA = (A_i)_{0\le i\le n}\in\Sigma^{n+1}$. The \emph{conditional probability distribution of $\mu$ given $\calA$} is defined as $\mu_\calA = \mu_{A_0,\ldots,A_{n}}$ by induction: $\mu_{A_0}$
	is defined as above, and for $1\le j\le n$:
	\[
		\mu_{A_0,\ldots,A_j} =
		\Omega_\calT(\mu_{A_0,\ldots,A_{j-1}})_{A_j}
		\puncteq{.}
	\]
	To be well-defined, we require in addition that $\mu(A_0)>0$ and that for
	all $0\le j\le n - 1$, $\Omega_\calT(\mu_{A_0,\ldots,A_j})(A_{j + 1}) > 0$.
\end{definition}
The intuition is that $\mu_\calA$ is the conditional distribution on $A_n$ after normalizing and restricting at each step $i$ the distribution $\mu$ to
the set $A_i$.

We connect the two interpretations of the semantics of an STS: as an
object generating a measure on infinite runs, and as a transformer of
probability distributions. The next result originates from~\cite[Lemma~5]{BBBC18}.
\begin{lemma} \label{lem:splitCyl}
	Let $\mu\in\Dist(S)$ be an initial distribution and $(A_i)_{0\le i\le n}$
	be a sequence of measurable sets of states. For $0\le j\le n$, we denote by
	$\mu_j$ the
	conditional distribution $\mu_{A_0,\ldots,A_j}$, which we assume to be
	well-defined.
	Then for every $0\le j\le n$, we have
	\begin{align*}
	\Prob^\calT_\mu(\Cyl(A_0,\ldots,A_n)) &=
	\Prob^\calT_\mu(\Cyl(A_0,\ldots,A_j))\cdot
	\Prob^\calT_{\Omega_\calT{(\mu_j)}}(\Cyl(A_{j+1},\ldots,A_n)) \\
	&= \Prob^\calT_\mu(\Cyl(A_0,\ldots,A_j))\cdot
	\Prob^\calT_{\mu_j}(\Cyl(A_j,\ldots,A_n))
	\puncteq{.}
	\end{align*}
\end{lemma}

\paragraph{Attractors}
We will be particularly interested in the existence of \emph{attractors} for STSs.

\begin{definition}[Attractor]
	A set $A\in\Sigma$ is an \emph{attractor for $\calT$} if for every
	$\mu\in\Dist(S)$, $\Prob^{\calT}_{\mu}(\F A)=1$.
\end{definition}

While $S$ is always an attractor for $\calT$, we will later search for
attractors with more interesting properties.
The definition of \emph{attractor} actually implies a seemingly
stronger statement: an attractor is almost surely visited
\emph{infinitely often} from any initial distribution.

\begin{lemma}[{\cite[Lemma~19]{BBBC18}}] \label{lem:AttrInfOften}
	Let $A$ be an attractor for $\calT$. Then, for every
	initial distribution $\mu\in\Dist(A)$, $\Prob^\calT_\mu(\GG\F
    A) = 1$.
\end{lemma}

\subsection{Decisiveness} \label{sec:dec}
Before introducing decisiveness,
we give the definition of an \emph{avoid-set}: for
$B \in \Sigma$, its \emph{avoid-set} is written as
$\Btilde = \{s \in S \mid \Prob_{\delta_s}^{\calT}(\F B) = 0\}$ (where $\delta_s$ is the Dirac distribution at $s$). The
avoid-set $\Btilde$ corresponds to the set of states from which
executions almost surely stay out of $B$ \textit{ad infinitum}. One
can show that the set $\widetilde{B}$ belongs to the $\sigma$-algebra
$\Sigma$~\cite[Lemma~14]{BBBC18}. We can now define the concept of
decisiveness as in~\cite{BBBC18}.

\begin{definition}[Decisiveness]
	Let $B\in\Sigma$ be a measurable set. We say that $\calT$ is \emph{decisive
	w.r.t.\ $B$} if for every $\mu\in\Dist(S)$,
	$\Prob^{\calT}_\mu(\F B \lor \F \Btilde) = 1$.
\end{definition}
Intuitively, the decisiveness property states that, almost surely,
either $B$ will eventually be visited, or states from which $B$ can no
longer be reached will eventually be visited.

\begin{example}[Random walk]
	We consider the random walk $\calT$ from Figure~\ref{fig:randomWalk}.
	We want to find out whether $\calT$ is decisive
	w.r.t.\ $B = \{0\}$. We assume that the initial distribution is given by
	$\delta_1$ (the Dirac distribution at $1$). By the theory on random walks,
	we know that if $\frac{1}{2} < p < 1$, the walk will almost surely diverge to
	$\infty$. This entails that $\Prob^\calT_{\delta_1}(\F B) < 1$ and
	$\Prob^\calT_{\delta_1}(\GG\F B) = 0$.
	Moreover, since $p < 1$, there is a path with
	positive probability from every state to $0$, so $\Btilde =	\emptyset$.
	Therefore,
	$
	\Prob^\calT_{\delta_1}(\F B \lor \F\Btilde) = \Prob^\calT_{\delta_1}(\F B)
	< 1\puncteq{,}
	$
	which means that $\calT$ is not decisive w.r.t.\ $B$. If $p\le\frac{1}{2}$, we have that
	$\Prob^\calT_{\delta_1}(\F B) = 1$. Hence, in this case, STS
	$\calT$ is decisive w.r.t.\ $B$.
	\begin{figure}[tbh]
		\centering
		\includegraphics[width=0.4\columnwidth]{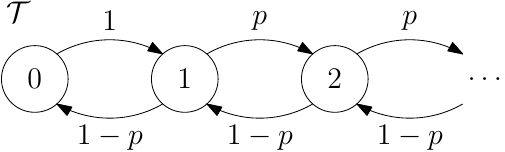}
		\caption{Random walk on $\IN$.}
		\label{fig:randomWalk}
	\end{figure}
\end{example}

A major interest of the decisiveness concept lies in the design of
simple procedures for the qualitative and quantitative analysis of stochastic
systems.
Indeed, as presented in~\cite{ABM07,BBBC18}, it allows for instance to compute,
under some effectiveness hypotheses, arbitrary close approximations of the
probabilities of various properties, like reachability, repeated
reachability, and even arbitrary $\omega$-regular properties. We refer
to~\cite[Sections~6 \&~7]{BBBC18} for more details, but briefly recall the
approximation scheme for reachability properties in order to illustrate the usefulness of the decisiveness property. This scheme will be applied to a specific class of STSs in Section~\ref{sec:quantReachDec}.

Let $B \in \Sigma$ be a measurable set and $\mu \in \Dist(S)$ be an
initial distribution. To compute an approximation of
$\Prob_\mu^\calT(\F B)$, we define two sequences $(\pnYes)_{n\in\IN}$
and $(\pnNo)_{n\in\IN}$ such that for $n\in\IN$,
\[
\pnYes = \Prob_\mu^\calT(\F[\leq n] B)\ \text{and}\
\pnNo = \Prob_\mu^\calT(\comp{B} \U[\leq n]{\Btilde})\puncteq{.}
\]
These sequences are non-decreasing and converge respectively to
$\Prob_\mu^\calT(\F B)$ and $\Prob_\mu^\calT(\comp{B} \U \Btilde)$.
Observe moreover that for all $n \in \IN$, we have that
\[
\pnYes \leq \Prob_\mu^\calT(\F B) \leq 1 - \pnNo\puncteq{.}
\]

The main idea behind decisiveness of STSs lies in the following
property~\cite{ABM07,BBBC18}: $\calT$ is decisive w.r.t.\ $B$ if and only if
\[
\lim_{n \to \infty} \pnYes + \pnNo = 1\puncteq{.}
\]
The equivalence comes from the fact that $\lim_{n \to \infty} \pnYes + \pnNo = \Prob^{\calT}_\mu(\F B \lor \F \Btilde)$~\cite[Proposition~105]{BBBC18}.
In situations where $\pnYes$ and $\pnNo$ can be effectively
approximated arbitrarily closely and $\calT$ is decisive w.r.t.\
$B$, we can thus approximate
$\Prob_\mu^\calT(\F B)$ up to any desired error bound.
Decisiveness is therefore \emph{exactly} the property required to guarantee the termination and correctness of this approximation scheme.

\subsection{A New Criterion for Decisiveness} \label{sec:mainDecCrit}
Our goal is to provide new sufficient conditions for
decisiveness of STSs. To this end, we expose the following crucial lemma.
For $\calA = (A_i)_{0\le i\le n}\in \Sigma^{n+1}$, we denote by
$\phi_\calA$ the \LTL formula
$A_0\land\X A_1\land\ldots\land\X[n] A_n$, where \X is the standard
``next'' modality, and \X[n] its $n$\textsuperscript{th} iterate. The missing
proofs for this section are provided in~\ref{sec:app-sts}.

\begin{restatable}{lemma}{GBcGFAGen} \label{lem:GBcGFAGen}
   	Let $B \in \Sigma$, and
    $\calA = (A_i)_{0\le i\le n}\in \Sigma^{n+1}$. Suppose that
    there is $p > 0$ such that for all $\nu \in \Dist(S)$ with $\nu_\calA$
    well-defined, we have
    $\Prob^\calT_{\nu_\calA}(\F B) \geq p$. Then for any
    $\mu \in\Dist(S)$,
	\[
	\Prob^\calT_\mu(\GG \comp{B} \land
	\GG\F \phi_\calA) = 0\puncteq{.}
	\]
\end{restatable}

This result seems rather intuitive: if we go infinitely often through the sequence $\calA$ in order, and after every passage through $\calA$ we have a probability bounded from below to reach $B$, then the probability to stay in $\comp{B}$ forever is $0$.
An equivalent statement for Markov chains has been used without proof in \cite[Lemmas~3.4 \&~3.7]{ABM07}.
A weaker version of this statement is given as part of the proof of~\cite[Proposition~36]{BBBC18}, where it is said that this general case was not known to be true or false.
This weaker version assumes that there is a uniform upper bound $k$ such that for all $\nu\in\Dist(S)$, $\Prob^\calT_{\nu_\calA}(\F[\le k] B) \geq p$ to obtain a similar conclusion.
We have removed the need for this constraint.

A possible proof of this result consists of regarding a stochastic transition system as a stochastic process, and resorting to results from martingale theory.
More precisely, using \textit{L\'evy's zero-one law}, we obtain that infinite runs that never reach $B$ are the same (up to a set of probability $0$) as the infinite runs $s_0s_1\ldots$ for which the probability to reach $B$ given $s_0\ldots s_n$ converges to $0$ as $n$ grows to infinity.
Runs that go through $\calA$ infinitely often cannot both avoid $B$ and have a probability to visit $B$ that converges to $0$ (since for every passage through $\calA$, the probability to visit $B$ is bounded from below by $p > 0$).
Therefore, such runs will almost surely visit $B$.

\paragraph{Main decisiveness criterion}
We can now state our main contribution to decisiveness.

\begin{restatable}[Decisiveness criterion]{proposition}{decCrit} \label{prop:mainDecCritGen}
  Let $B\in\Sigma$ and
  $m\in\IN\setminus\{0\}\cup\{\infty\}$. For every $0 \le j < m$, let $n_j \in \IN$
  and $\calA_j = (A^{(j)}_i)_{0\le i\le n_j} \in \Sigma^{n_j+1}$. We
  assume that for all $\nu\in\Dist(S)$,
	\[
	\Prob^\calT_\nu\Big(\bigvee_{0\le j< m}\GG\F \phi_{\calA_j}\Big) = 1
	\puncteq{.}
	\]
	For $\nu\in\Dist(S)$, when well-defined, we set $\nu_j^* = (\nu_{\calA_j})_{\comp{(\Btilde)}}$ for the conditional distribution of $\nu$ given sequence $\calA_j$, which is then conditioned on $\comp{(\Btilde)}$ (as defined in Definition~\ref{def:condDist}).
	Assume that there exists $p > 0$ such that for all $0\le j< m$,
	for all $\nu\in\Dist(S)$ such that $\nu_{\calA_j}$ is well-defined, either
	$\nu_{\calA_j}(\Btilde) = 1$ or
	\[
	\Prob_{\nu_j^*}^\calT(\F B)\ge p
	\puncteq{.}
	\]
	Then $\calT$ is decisive w.r.t.\ $B$.
\end{restatable}

In order to provide an intuitive understanding of this proposition, we
instantiate its statement with $m = 1$, $n_0 = 0$.

\begin{corollary} \label{cor:mainDecCrit}
  Let $B\in\Sigma$ be a
  measurable set, and $A\in\Sigma$ be an attractor for $\calT$. We
  denote by $A' = A\cap\comp{(\Btilde)}$ the set of states of $A$ from
  which $B$ is reachable with a positive probability. Assume that
  there exists $p > 0$ such that for all $\nu\in\Dist(A')$,
  $\Prob_\nu^\calT(\F B)\ge p$. Then $\calT$ is decisive w.r.t.\ $B$.
\end{corollary}
With probability $1$, every run visits attractor $A$ infinitely
often (Lemma~\ref{lem:AttrInfOften}), but the hypotheses imply a dichotomy
between runs. Some runs will reach a state of $A$ from which $B$ is
almost surely non-reachable, and will end up in $\Btilde$. The other
runs will go infinitely often through states of $A$ such that the
probability of reaching $B$ is lower bounded by $p$ (i.e., states of
$A'$), and will almost surely visit $B$ by Lemma~\ref{lem:GBcGFAGen}.
This almost-sure dichotomy between runs is required to show
decisiveness. In the more general statement of
Proposition~\ref{prop:mainDecCritGen}, instead of a simple attractor
$A$, we assume that we visit infinitely often some finite sequences of
sets of states. This allows us to have a weaker assumption on the
probability lower bound $p$; it is enough to obtain this lower bound
from more specific conditional distributions.
This generalization will be crucial when applying this criterion to a specific class of STSs in Section~\ref{sec:decisiveSHS}.

This criterion strictly generalizes those used in the
literature:~\cite[Lemmas~3.4 \&~3.7]{ABM07} and~\cite[Propositions~36
\&~37]{BBBC18} (see proofs page~\pageref{app:prop:gen1}). The criterion
in~\cite[Lemma~3.4]{ABM07} assumes the existence of a finite attractor; the
criteria in~\cite[Propositions~36 \&~37]{BBBC18} assume some finiteness
property in an abstraction (see next section), which we do not.
In~\cite[Lemma~3.7]{ABM07}, a similar kind of property as ours is required from
all the states of the STSs, not only from some specific distributions.

\subsection{Abstractions of Stochastic Transitions Systems} \label{sec:abstractions}
Decisiveness and abstractions are deeply intertwined concepts, so we
briefly recall this notion~\cite{BBBC18} and related properties. We
let $\calT_1 = (S_1,\Sigma_1,\kappa_1)$ and
$\calT_2 = (S_2,\Sigma_2,\kappa_2)$ be two STSs, and
$\alpha\colon (S_1,\Sigma_1) \to (S_2,\Sigma_2)$ be a measurable
function. We say that a set $B \in \Sigma_1$ is \emph{$\alpha$-closed}
if $B = \alpha^{-1}(\alpha(B))$ (in other words, if for all $s \notin B$, $\alpha(s) \notin \alpha(B)$). To mean that $B$ is
$\alpha$-closed, we also say that $\alpha$ is \emph{compatible} with $B$.
Following~\cite{Bog07,GBK16}, we
define a natural way to transfer measures from $(S_1,\Sigma_1)$ to
$(S_2,\Sigma_2)$ through $\alpha$: the \emph{pushforward of $\alpha$}
is the function $\alpha_\# \colon \Dist(S_1) \to \Dist(S_2)$ such that
$ \alpha _\#(\mu)(B) = \mu(\alpha^{-1}(B)) $ for every
$\mu \in\Dist(S_1)$ and for every $B \in \Sigma_2$.

\begin{definition}[$\alpha$-abstraction]
	STS $\calT_2$ is an \emph{$\alpha$-abstraction} of $\calT_1$ if for
	all $\mu \in \Dist(S_1)$,
	\[
	\alpha_{\#}(\Omega_{\calT_1}(\mu))\ \text{is
	qualitatively equivalent to} \ \Omega_{\calT_2}(\alpha_{\#}(\mu))
	\puncteq{.}
	\]
\end{definition}
Informally, the two STSs have the same ``qualitative'' single steps.
Later, one may speak of \emph{abstraction} instead of
$\alpha$-abstraction if $\alpha$ is clear in the context.

Notice that $\alpha$ induces a natural equivalence relation $\sim_\alpha$ on
$S_1$: for $s,s'\in S_1$, $s\sim_\alpha s'$ if and only if $\alpha(s) = \alpha(s')$.
When the index of $\sim_\alpha$ (i.e., its number of equivalence classes) is countable, we can characterize when $\alpha$ can be used to define an \emph{$\alpha$-abstraction}.

\begin{lemma}
	If $\sim_\alpha$ has countable index,
	there is an $\alpha$-abstraction of $\calT_1$ if and only if
	for all $P\in \quotient{S_1}{\sim_\alpha}$, $\Pre^{\calT_1}(P)$ is
	$\alpha$-closed.
\end{lemma}
As defined earlier, $\Pre^{\calT_1}(P)$ refers to the set of states of $\calT_1$ that can reach $P$ in one step with a \emph{positive probability}. This contrasts with the classical definition of the operator $\Pre(P)$ for non-deterministic transition systems, which usually refers to the set of states from which there is a transition reaching a state in $P$.
This lemma suggests that abstractions of $\calT_1$ can be seen as
partitions that are stable w.r.t.\ the function $\Pre^{\calT_1}$, to which we
only need to add stochastic transitions between the adequate pieces of the
partition.
This also indicates that to obtain an $\alpha$-abstraction from a finite
partition, we can apply a procedure very similar to the classical
bisimulation algorithm for
non-deterministic transition systems, adapting slightly the meaning of
operator $\Pre$ to our stochastic setting. The procedure is given in
Procedure~\ref{alg:abs}. We start with an initial
(finite) partition $\calP_{\textit{init}}$ of the state space, of which we want
a compatible abstraction, i.e., an abstraction compatible with every set of states in $\calP_{\textit{init}}$.

\renewcommand{\algorithmicrequire}{\textbf{Inputs: }}
\renewcommand{\algorithmicensure}{\textbf{Output: }}

\begin{algorithm}\caption{Abstraction refinement procedure.} \label{alg:abs}
	\begin{flushleft}
	\algorithmicrequire{an STS $\calT$, and an initial finite partition
		$\calP_{\textit{init}}$ of its state space.}\\
	\algorithmicensure{the coarsest finite abstraction compatible with
	$\calP_{\textit{init}}$, if it exists.}
	\end{flushleft}
	\begin{algorithmic}
		\STATE $\calP\gets \calP_{\textit{init}}$
		\WHILE{$\exists P, P' \in \calP\ \text{such that}\
		P'\cap\Pre^\calT(P)\neq\emptyset\land
			P'\setminus\Pre^\calT(P)\neq\emptyset$}
		\STATE $P_1\gets P'\cap\Pre^\calT(P), P_2\gets P'\setminus\Pre^\calT(P)$
		\STATE $\calP\gets (\calP\setminus\{P'\})\cup\{P_1,P_2\}$
		\ENDWHILE
		\RETURN $\calP$
	\end{algorithmic}
\end{algorithm}

\begin{lemma}
Procedure~\ref{alg:abs} terminates if and only if there exists a finite
abstraction of $\calT_1$ compatible with $\calP_{\textit{init}}$. In this case,
it returns the coarsest such partition.
\end{lemma}

The objective behind the notion of abstraction is that by finding an
$\alpha$-abstraction $\calT_2$ which is somehow simpler than $\calT_1$
(for example, with a smaller state space), we should be able to use
$\calT_2$ (with initial distribution $\alpha_\#(\mu)$) to analyze some
properties of $\calT_1$ (with initial distribution $\mu$). To do so,
we need to know which properties are preserved through an
$\alpha$-abstraction. As a first observation, positive probability of
reachability properties is preserved~\cite[Corollary~94]{BBBC18}.
In particular, avoid sets are preserved: for $B\in\Sigma_2$, $\alpha^{-1}(\Btilde) = \widetilde{\alpha^{-1}(B)}$.
Stronger conditions are
required to study almost-sure reachability properties of an STS
through its $\alpha$-abstraction. We select a definition and two key
results of~\cite{BBBC18} about that matter which will be useful in the
subsequent sections.

\begin{definition}[Sound $\alpha$-abstraction]
  We say that $\calT_2$ is a \emph{sound} $\alpha$-abstraction of $\calT_1$ if
  for all $B\in\Sigma_2$, $\Prob^{\calT_2}_{\alpha_\#(\mu)}(\F B) = 1$
  implies $\Prob^{\calT_1}_\mu(\F \alpha^{-1}(B)) = 1$.
\end{definition}

Sound abstractions preserve almost-sure reachability properties from $\calT_2$
to $\calT_1$ (but not necessarily the converse). Soundness of
the abstraction is a sufficient condition to lift
decisiveness of the abstraction to the original system: for every
$B\in\Sigma_2$, if $\calT_2$ is
decisive w.r.t.\ $B$, then $\calT_1$ is decisive w.r.t.\ $\alpha^{-1}(B)$
\cite[Proposition~33]{BBBC18}.
We also have that some decisiveness property is sufficient to
guarantee soundness of the abstraction.

\begin{proposition}[{\cite[Prop.~40]{BBBC18}}]
	\label{prop:decImpliesSound}
	If $\calT_1$ is decisive w.r.t.\ every $\alpha$-closed
	set, then $\calT_2$ is a sound $\alpha$-abstraction of $\calT_1$.
\end{proposition}

We can now formulate the most important result from this section, linking the
qualitative reachability of an STS and its abstraction under soundness and
decisiveness hypotheses.
\begin{proposition}[{\cite[Prop.~6.1.9]{Car17}}] \label{prop:almostSureReach}
Let $B\in\Sigma_2$ be a measurable set of $\calT_2$. If $\calT_2$ is a
sound $\alpha$-abstraction of $\calT_1$ and is decisive w.r.t.\ $B$,
then for every $\mu\in\Dist(S_1)$,
$\Prob_\mu^{\calT_1}(\F \alpha^{-1}(B)) = 1$ if and only if
$\Prob_{\alpha_{\#}(\mu)}^{\calT_2}(\F B) = 1$.
\end{proposition}

	\section{Stochastic Hybrid Systems}
		\label{sec:stochHS}
		From now on, we focus on a subclass of stochastic transition systems that consists in a stochastic extension of the well-studied \emph{hybrid systems}.
A \emph{hybrid system} is a dynamical system combining discrete and continuous
transitions. It can be defined as a non-deterministic automaton with a finite
number of continuous variables, whose evolution is described via an infinite
transition system. Hybrid systems have been widely studied since their
introduction in the 1990s (e.g.,~\cite{ACH+,Hen96}). They are
effectively used to model various time-dependent reactive systems; systems that
need to take into account both continuous factors (e.g., speed, heat, time,
distance) and discrete factors (e.g., events, instructions) are ubiquitous.

It is worthwhile to add \emph{stochasticity} to hybrid systems, as this permits
a more fine-grained analysis by distinguishing scenarios that are likely to
happen and scenarios that are not. If we cannot
completely prevent an undesirable outcome from happening, it is still
beneficial to have the ability to quantify its probability.

We define hybrid systems and give them a fully stochastic semantics, yielding the
class of \emph{stochastic hybrid systems} (SHSs).
We then prove in Section~\ref{sec:undecSHS} the undecidability of reachability
problems for SHSs, thereby showing the interest of finding large decidable
subclasses. This undecidability result is not surprising, as these problems are
already undecidable in non-deterministic hybrid systems, but a proof in the
stochastic context allows to reason about the quantitative problem as well.

\subsection{Hybrid Systems}
We proceed with the definition of a (non-deterministic) \emph{hybrid system}.
\begin{definition}[Hybrid system]
	A \emph{hybrid system} (HS) is a tuple
	$\calH=(L,X,\calA,E,\flow,\inv,\calG,\calR)$
	where:
	$L$ is a finite set of \emph{locations} (discrete states);
	$X=\{x_1,\ldots,x_n\}$ is a finite set of \emph{continuous  variables};
	$\calA$ is a finite alphabet of \emph{events};
	$E \subseteq L \times \calA \times L$ is a finite set of \emph{edges};
	for each $\ell \in L$, $\flow(\ell)\colon \IR^n\times \IR^+ \to \IR^n$
	is a continuous function describing the \emph{dynamics} in location $\ell$;
	$\inv$ assigns to each location a subset of $\IR^{n}$ called
	\emph{invariant};
	$\calG$ assigns to each edge a subset of $\IR^{n}$ called \emph{guard};
	$\calR$ assigns to each edge $e$ and valuation $\ve\in\IR^n$ a subset
	$\calR(e)(\ve)$ of $\IR^n$ called \emph{reset}.
	For $\ell\in L$, $e\in E$, we usually denote $\flow(\ell)$ and $\calR(e)$
	by $\flow_\ell$ and $\calR_e$.
\end{definition}

We denote the number of variables $|X|$ by $n$.
We denote by $\IR^X$ the set of valuations that map variables from $X$ to real
numbers.
In what follows, we treat
elements of $\IR^X$ as elements of $\IR^n$ through the bijection $\ve
\mapsto (\ve(x_1),\ldots,\ve(x_n))$.

We now give the semantics of hybrid systems.
Given a hybrid system $\calH$, we define $S_\calH = L\times\IR^n$ as the states
of $\calH$.
We distinguish two kinds of transitions between states:
\begin{itemize}
	\item there is a \emph{switch-transition}
	$(\ell,\ve)\xrightarrow{a}(\ell',\ve')$ if there exists
	$e = (\ell,a,\ell')\in E$
	such that $\ve\in \inv(\ell)\cap\calG(e)$, $\ve'\in \calR_e(\ve)\cap
	\inv(\ell')$;
	\item there is a \emph{delay-transition}
	$(\ell,\ve)\xrightarrow{\tau}(\ell,\ve')$ if there exists $\tau\in\IR^+$
	such that for all $0 \leq \tau' \leq \tau$, $\flow_\ell(\ve,\tau') \in
	\inv(\ell)$ and $\ve' = \flow_\ell(\ve,\tau)$.
\end{itemize}
Informally, a switch-transition $(\ell,\ve)\xrightarrow{a}(\ell',\ve')$ means
that an edge
$e=(\ell, a, \ell')$ can be taken without violating any constraint: the value
$\ve$ of the continuous variables is an element of the invariant $\inv(\ell)$
and of the guard $\calG(e)$,
and there is a possible reset $\ve'$ of the variables which is an element of
the invariant $\inv(\ell')$. A delay-transition
$(\ell,\ve)\xrightarrow{\tau}(\ell,\ve')$ means that some time $\tau$ elapses
without changing the discrete location of the system---the only constraint is
that all the values taken by the continuous variables during this time are in
the invariant $\inv(\ell)$.

Given $s = (\ell, \ve)\in L \times \IR^n$ a state of
the hybrid system, and $\tau \in \IR^+$, we denote by $s + \tau = (\ell,
\flow_\ell(\ve, \tau))$ the new state after time $\tau$ has elapsed, without
changing the location; $\tau$ is then referred to as a \emph{delay}.

We only consider \emph{mixed transitions} in runs, i.e., transitions that
consist of a delay-transition (some time elapses) followed by a
switch-transition (an edge is taken and the location changes).
A mixed transition is denoted by $(\ell,\ve)\xrightarrow{\tau, a}(\ell', \ve')$
if and only if there exists $\ve''\in\IR^n$ such that
$(\ell,\ve)\xrightarrow{\tau}(\ell, \ve'')\xrightarrow{a}(\ell',\ve')$.

We usually assume that there is a bijection between the edges $E$ and the
alphabet of events $\calA$, and we omit mentioning this alphabet. If $e =
(\ell, a, \ell')\in E$, we can thus denote $\xrightarrow{e}$ (resp.\
$\xrightarrow{\tau, e}$) for switch-transitions (resp.\ mixed transitions)
instead of $\xrightarrow{a}$ (resp.\ $\xrightarrow{\tau, a}$), if there is no
ambiguity.

\begin{definition}[Run of an HS]
	A \emph{run} of an HS is an infinite sequence
	\[
	(\ell_0, \ve_0)\xrightarrow{\tau_0, e_0} (\ell_1, \ve_1)
	\xrightarrow{\tau_1, e_1} (\ell_2, \ve_2) \xrightarrow{\tau_2, e_2}\ldots
	\]
	of elements in $L\times\IR^n$ such that for all $i\ge 0$,
	$(\ell_i, \ve_i)\xrightarrow{\tau_i, e_i} (\ell_{i+1}, \ve_{i+1})$ is a
	mixed transition.
\end{definition}

\begin{example} \label{ex:pacman1}
	We provide in Figure~\ref{fig:pacman} an example of a hybrid system.
	This example was first studied in~\cite{BBBBG08}. There are two continuous
	variables ($x$ and $y$) and five
	locations, each of them equipped with the same simple dynamics: $\dot{x} =
	\dot{y} = 1$ (i.e., $\gamma_\ell((x,y), \tau) = (x + \tau, y + \tau)$ for
	every location $\ell\in L$). Locations $\ell_2$ and $\ell_4$ have the same
	invariant,
	which is $\{(x,y)\mid y < 1\}$; the other invariants are simply $\IR^2$.
	Guards are written next to the edge to which they are
	related: for instance, $\calG(e_4) = \{(x,y)\mid y = 2\}$. The notation
	``$x\defeq 0$'' is used to denote a deterministic
	reset (in this case, the value of $x$ is reset to $0$ after taking
	the edge). For instance, $\calR_{e_1}(x, y)= \{(x, 0)\}$ (the value of
	$x$ is preserved and $y$ is reset to $0$).
	If nothing else is written next to an edge $e$, it means that there
	is no reset on $e$, i.e., that $\calR_{e}(\ve)= \{\ve\}$ for all $\ve\in\IR^n$.
	An example of the beginning of a run of this system can be
	\begin{align*}
		(\ell_0, (0,0))&\xrightarrow{0.4, e_0} (\ell_1, (0.4, 0.4))
		\xrightarrow{0.6, e_1} (\ell_2, (1,0))\xrightarrow{0.2,e_2}(\ell_0,
		(0,0.2)) \\
		&\xrightarrow{1.5, e_3} (\ell_3, (1.5,1.7))\xrightarrow{0.3,e_4}
		(\ell_4,(1.8,0))\xrightarrow{0.8, e_5}(\ell_0,(0, 0.8))\ldots
	\end{align*}
	Due to his fairly simple dynamics, guards and resets, this hybrid
	system actually belongs to the class of \emph{timed automata}~\cite{AD94}.
	\begin{figure}[tbh]
		\centering
		\includegraphics[width=0.6\columnwidth]{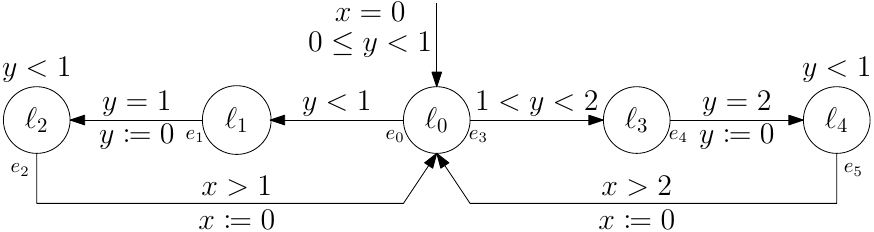}
		\caption{Example of a hybrid system with two continuous variables. Each
		location is equipped with the dynamics $\dot{x} = \dot{y} = 1$.}
		\label{fig:pacman}
	\end{figure}
\end{example}

We now give more vocabulary to refer to hybrid systems.
If there is a switch-transition $(\ell, \ve) \xrightarrow{e} (\ell', \ve')$, we
say that edge $e$ is \emph{enabled} at state $(\ell, \ve)$.
An edge $e$ is enabled at state $(\ell, \ve)$ if it can be taken with no delay from state $(\ell,\ve)$.
Given a state $s = (\ell, \ve)$ and an edge $e = (\ell, a, \ell')$ of
$\calH$, we define $I(s, e) = \{\tau \in \IR^+ \mid s \xrightarrow{\tau, e}
s'\}$ as the set of delays after which edge $e$ is enabled from $s$, and $I(s)
= \bigcup_e I(s, e)$ as the set of delays after which any edge is
enabled from $s$.
For instance,
in the hybrid system from Example~\ref{ex:pacman1}, for $s =
(\ell_0,(0,0.2))$, the set $I(s,e_0) = \intervalco{0,0.8}$, and $I(s) =
\intervalco{0,1.8}\setminus\{0.8\}$.

We say that a state $s \in L \times \IR^n$ is \emph{non-blocking} if $I(s) \neq \emptyset$.
In the sequel, we only consider hybrid systems such that all states are
non-blocking, which implies that a mixed transition is always possible from any state.
When adding stochasticity on top of this model, this will allow to have a well-defined probability distribution on runs, which consist of mixed transitions.

\subsection{Probabilistic Semantics for Hybrid Systems}
We expand on the definition of a hybrid system by replacing the
non-deterministic aspects of the definition with stochasticity.
\emph{Stochastic hybrid systems} will be our main focus of attention
in the rest of the paper.
\begin{definition}[Stochastic hybrid system]
	A \emph{stochastic hybrid system} (SHS) is defined as a tuple $\calH=(\calH',\delDist_L,\resDist_\calR,\edgDist)$, where:
	\begin{itemize}
		\item $\calH'=(L,X,\calA,E,\flow,\inv,\calG,\calR)$
		is a hybrid system,
		which is referred to as the \emph{underlying hybrid system of $\calH$}.
		We require guards, invariants and resets to be Borel sets.
		\item $\delDist_L \colon L \times \IR^n \to \Dist(\IR^+)$ associates to each state a probability
		distribution on the time delay in $\IR^+$ (equipped with the
		classical Borel $\sigma$-algebra) before leaving a location.
		Given $s \in L \times \IR^n$ (a state of $\calH$), the distribution
		$\delDist_L(s)$ will also be denoted by $\delDist_s$. We require that
		for every $s \in L \times \IR^n$, $\delDist_s(\IR^+) = \delDist_s(I(s))
		= 1$, i.e., the probability that at least one edge is enabled after a
		delay is $1$.
		\item $\resDist_\calR$ associates to each edge $e$ and valuation $\ve$
		a probability distribution on the set $\calR_e(\ve)\subseteq \IR^n$.
		Given $e\in E$ and $\ve\in\IR^n$, the
		distribution $\resDist_\calR(e)(\ve)$ will also be denoted by
		$\resDist_e(\ve)$.
		\item $\edgDist\colon L\times\IR^n \to \Dist(E)$ is a function
		assigning to each state of $\calH$ a probability distribution on the
		edges. We require that $\edgDist(s)(e) > 0$ \emph{if and only if} edge
		$e$ is enabled at $s$. For $s\in L\times\IR$, we denote $\edgDist(s)$
		by $\edgDist_s$. This distribution is only defined for states at which
		at least one edge is enabled.
	\end{itemize}
\end{definition}

\begin{remark}
	The term ``stochastic hybrid system'' is used for a wide variety of
	stochastic extensions of hybrid systems throughout the literature.
	In this work, we consider stochastic delays, stochastic
	resets, a stochastic edge choice, and initial states given by a probability
	distribution.
	The way this probabilistic semantics is added on top of hybrid systems
	is very similar to how \emph{timed automata} are converted to
	\emph{stochastic timed automata} in~\cite{BBBBG08,BBBM08,BBB+14}. The
	major difference is that in the case of timed automata, resets
	are deterministic, hence do not require stochasticity.
	One of the differences compared to many SHS models from the literature is that our dynamics are deterministic: after a stochastic reset has been performed in a location, the dynamics in this location are not affected by randomness (as opposed to dynamics described by stochastic differential equations; see, e.g.,~\cite{HLS00,Bu12}).
	This restriction is due to the fact that we will assume the definability of the dynamics in some mathematical structure in Section~\ref{sec:ominimal}, which is difficult to determine in general when the dynamics are given by differential equations.

	Although dynamics are deterministic, the model is
	powerful enough to
	emulate some randomness in dynamics by assuming that extra variables are
	used to influence the continuous flow of the other variables. These
	variables can be chosen stochastically in each location through the reset
	mechanism.
	This is for example sufficient to consider a class of models that could constitute a potential stochastic extension of \emph{rectangular automata}~\cite{HKPV}, whose
	variables evolve according to slopes inside an interval; the actual value of each slope would then
	be given by one of these extra variables and could be reset in each location. For example, we could have two variables $x$ and $m_x$ such that $m_x$ is reset with a uniform distribution on $\intervalcc{1,4}$ and $\dot{x} = m_x$.
	Our model is very similar to the underlying theoretical model of the software tool \textsc{Uppaal}~\cite{DDL+12}.
	In Section~\ref{sec:ominimal}, we will identify the need to restrict some more the definition of some components of SHSs to ensure their definability; this is however not required for the results of Section~\ref{sec:decisiveSHS}.
\end{remark}

When referring to an SHS, we make in particular use of the same
terminology as for hybrid systems (e.g., runs, enabled edges, allowed delays
$I(\cdot)$) to describe its underlying hybrid system.

In order to apply the theory developed in Section~\ref{sec:sts}, we
give the semantics of an SHS $\calH$ as an STS $\calT_{\calH} =
(S_{\calH},\Sigma_{\calH},\kappa_{\calH})$. The set $S_{\calH}$ is the
set $L \times \IR^n$ of states of $\calH$, and $\Sigma_{\calH}$ is the
$\sigma$-algebra product between $2^L$ and the Borel $\sigma$-algebra
on $\IR^n$.  To define $\kappa_\calH$, we first explain briefly the
role of each probability distribution in the definition of SHS.
Starting from a state $s = (\ell, \ve)$, a delay $\tau$ is chosen
randomly, according to the distribution $\delDist_s$. From the state
$s+\tau = (\ell,\ve')$, an edge $e=(\ell,a,\ell')$ (enabled in
$s+\tau$) is selected, following the distribution $\edgDist_{s + \tau}$
(such an edge is almost surely available, as $\delDist_s(I(s)) = 1$ by
hypothesis). The next state will be in location $\ell'$, and the
values of the continuous variables are stochastically reset according
to the distribution $\resDist_e(\ve')$.  We can thus define
$\kappa_{\calH}$ as follows: for $s = (\ell, \ve) \in S_\calH$, $B \in
\Sigma_\calH$,
\[
\kappa_\calH(s,B)
= \int_{\tau \in \IR^+}
\sum_{e = (\ell,a,\ell')\in E} \left(
\edgDist_{s+\tau}(e) \cdot
\int_{\ve'' \in \IR^n}
\ind{B}(\ell',\ve'') \ud (\resDist_e(\flow_\ell(\ve, \tau)))(\ve'')
\right) \ud \delDist_s(\tau)
\]
where $\ind{B}$ is the characteristic function of $B$. It gives the
probability to hit set $B \subseteq S_{\calH}$ from state $s$ in one step
(representing a mixed transition).
The function $\kappa_{\calH}(s, \cdot)$ defines a probability
distribution for all $s\in S_\calH$.

\begin{definition}[STS induced by an SHS]
	For an SHS $\calH$, we define $\calT_\calH = (S_\calH, \Sigma_\calH,
	\kappa_\calH)$ as the \emph{STS induced by $\calH$}.
\end{definition}

Thanks to the stochasticity of our models, we can reason about
both \emph{qualitative} and \emph{quantitative} reachability problems, as
defined in Definition~\ref{def:problems}.

\subsection{Undecidability of Reachability for Stochastic Hybrid Systems}
\label{sec:undecSHS}
We now show that qualitative and quantitative reachability problems for SHSs
are undecidable, even for SHSs with relatively simple features.
This demonstrates the significance of establishing results about the
decisiveness of classes of SHSs (done in Section~\ref{sec:decisiveSHS}).
Along with classical effectiveness assumptions (in
Section~\ref{sec:ominimal}), these decisiveness results will be sufficient to
guarantee the decidability of reachability problems for these classes.

Reachability problems have been extensively studied for
non-deterministic hybrid systems, and some of the undecidability
proofs~\cite{HKPV,HR00} can be translated almost directly to our
stochastic setting.  A classical method to show undecidability
consists in encoding every computation of a Turing-complete
model as an execution of a hybrid system.  The undecidability proof
of~\cite{HKPV} builds for every Turing-machine $M$ a hybrid system
with one accepting run that encodes the halting computation of
$M$. This proof is not ``robust'' in the sense that a slight
perturbation of this accepting run does not encode an execution of
$M$. As argued in~\cite{HR00}, this means that undecidability might
stem from the perfect (unrealistic) accuracy required to process such
an execution, and not from the very nature of hybrid systems. This is why
in~\cite{HR00}, the authors establish an undecidability result for
\emph{robust} hybrid systems, i.e., hybrid systems such that if they
accept some run $\run$, have to accept all runs ``close enough'' to
$\run$. This provides a more convincing argument that hybrid systems
are intrinsically undecidable.

In our stochastic framework, to obtain a similar idea of ``robustness'', our
goal is to prove undecidability even when constrained to purely
continuous distributions on time delays from any state, and very simple guards,
resets, and dynamics.
This requires a distinct proof from~\cite{HKPV}.
The result from~\cite{HR00} is much closer to the one that we want to
achieve, and we take inspiration from its proof to show the undecidability of
SHSs.
The proof consists of reducing the \emph{halting problem for
two-counter machines} to deciding whether a measurable
set in an SHS is reached with probability $1$. It is provided in~\ref{sec:app-shs}. We obtain the following result.

\begin{restatable}{proposition}{undecSHSs} \label{prop:undecSHS}
	The qualitative reachability problems and the approximate
	quantitative reachability problem are undecidable for stochastic
	hybrid systems with purely continuous distributions on time delays, guards
	that are linear comparisons of variables and constants, and using
	positive integer slopes for the flow of the continuous variables. The
	approximate quantitative problem is undecidable for any fixed
	precision $\epsilon < \frac{1}{2}$.
\end{restatable}

Although the proof is centered on showing the undecidability
of qualitative reachability problems, we get as a by-product the undecidability
of the approximation. Indeed, as the systems used throughout the proof reach
a target set $B$ with a probability that is either $0$ or $1$, the ability to
approximate $\Prob^{\calT_\calH}_\mu(\F B)$ with $\epsilon < \frac{1}{2}$ would
be sufficient to solve the qualitative problem. As the proof shows that these
qualitative problems are undecidable, we obtain that the approximate
quantitative problem is also undecidable. The proof therefore also shows that  deciding whether a state lies in $\Btilde$ is already undecidable. By considering a stochastic framework
rather than a non-deterministic one as in~\cite{HR00}, we thus obtain a
seemingly more powerful result with a similar proof.

	\section{Properties of Cycle-Reset Stochastic Hybrid Systems}
		\label{sec:decisiveSHS}
		The literature about non-deterministic hybrid systems suggests that to obtain
subclasses for which
the reachability problem becomes decidable, one must set sharp restrictions on
the continuous flow of the variables, and/or on the discrete transitions (via
the \emph{reset} mechanism).
In this decidable spectrum lie for instance
the \emph{rectangular initialized automata}, which are quite permissive toward
the continuous evolution of the variables, but need strong hypotheses about
what is allowed in the discrete transitions~\cite{HKPV}.
Our approach lies at one end of this spectrum: we will simply restrict the
discrete behavior by considering \emph{strong resets}, i.e., resets that forget
about the previous values of the variables, decoupling the discrete behavior
from the continuous behavior. We show that one strong reset per
cycle of the graph is sufficient to obtain our results, and we name this
property \emph{\cyclereset{}}.
This point of view has already been studied in~\cite{BMRT04,Gen05,LPS00} for
non-deterministic hybrid systems.

\begin{definition}[Strong reset]
	Given $\calH$ a stochastic hybrid system and $e$ an edge of $\calH$, we
	say that $e$ has a \emph{strong reset} (or is \emph{strongly reset}) if there exist $\calR_e^* \subseteq \IR^n$ and $\resDist_e^* \in \Dist(\calR_e^*)$ such that for all $\ve \in \calG(e)$, $\calR_e(\ve) = \calR_e^*$ and $\resDist_e(\ve) = \resDist_e^*$.
\end{definition}

If an edge $e$ is strongly reset, it stochastically resets all the continuous
variables when it is taken, and the stochastic reset does not depend on their
values.
It means that for any $\ve, \ve'\in\calG(e)$, the distributions
$\resDist_e(\ve)$ and $\resDist_e(\ve')$ are equal. If an edge is strongly
reset, we denote its reset distribution by $\resDist_e^*$ (which equals
$\resDist_e(\ve)$ for all $\ve\in\IR^n$).

Let $\calH = (L, \calA, E, \ldots)$ be an SHS. We denote by
\begin{align*}
C^\calH = \{ (e_0, e_1, \ldots, e_m) \in E^{m + 1} \mid\
&\forall 0 \leq i \leq m, e_i = (\ell_i, a_i, \ell_i'),\\
& \ell'_{m} = \ell_0, \forall 0 \leq i < m, \ell_i' = \ell_{i+1},\\
&\text{and } \forall 0 \leq i < j \leq m, e_i \neq e_j\}
\end{align*}
the set of \emph{simple cycles} of $\calH$.

\begin{definition}[\Cyclereset{} SHSs]
	We say that an SHS $\calH$ is \emph{\cyclereset{}} if for every simple
	cycle $(e_0,\ldots,e_m)\in C^\calH$, there exists $0\le i\le m$ such
	that $e_i$ is strongly reset.
\end{definition}

We show two independent and very convenient results of \cyclereset{} SHSs: such
SHSs are \emph{decisive w.r.t.\ any measurable set} (the proof of this
statement relies on the decisiveness criterion from
Proposition~\ref{prop:mainDecCritGen}), and \emph{admit a finite
abstraction}.

\begin{remark}
	It is interesting to notice that properties similar to ``one strong reset
	per cycle'' are given in various places throughout the literature about
	timed and hybrid systems, which may be due to technical usefulness. In~\cite{BBBM08}, the authors perform the
	quantitative analysis of stochastic timed automata with only one continuous
	variable, assuming that any bounded cycle in an abstraction contains a
	reset of the continuous variable. In~\cite{Gen05}, the class of
	\emph{relaxed o-minimal (non-deterministic) hybrid systems}, with one strong reset per cycle,
	is shown to admit a finite bisimulation, making the reachability problem
	decidable in cases where this bisimulation is effectively computable.
\end{remark}

\subsection{Decisiveness}

We motivate this section with an example of a simple non-decisive SHS, which we
will use to show that our decisiveness result is tight.

\begin{example} \label{ex:stochPacman}
	We add a stochastic layer to the hybrid system of
	Example~\ref{ex:pacman1}, pictured in Figure~\ref{fig:pacman}. The
	distributions on the time delays in locations $\ell_0$, $\ell_2$ and
	$\ell_4$
	are uniform distributions on the interval of allowed delays. For instance,
	at state $s = (\ell_0, (x, y))$, the distribution $\delDist_s$ follows a
	uniform
	distribution $\calU(0, 2 - y)$. In locations $\ell_1$ and $\ell_3$, the
	distributions on the delays are Dirac distributions, since the sets of
	allowed delays are singletons. As all resets are deterministic, they are
	simply modeled as Dirac distributions and since at most one edge is enabled
	in each state, the distributions $\edgDist_s$ are also Dirac distributions.

	It is proved in~\cite[Section~6.2.2]{BBB+14} that this SHS is not decisive
	w.r.t.\ $B = \{\ell_2\}\times\IR^2$.
	The reason is that from a state $s = (\ell_0, (0, y))$ with $0\le y < 1$,
	at each subsequent passage through location $\ell_0$, the value of $y$
	increases but stays bounded from above by $1$,
    which decreases the probability to take edge $e_0$ (and thus reach $B$).
    Let $\mu$ be the Dirac distribution at
	$(\ell_0, (0,0))$. As we will go infinitely often from $\mu$ through
	$\{\ell_0\}\times \{0\}\times\intervalco{0,1}$, we will never reach a state
	in $\Btilde$. However, it is proved that
	$\Prob_\mu^{\calT_\calH}(\GG\comp{B}) > 0$.
	The proof is quite technical and we do not recall it here.
	This implies in particular that
	\[
		\Prob_\mu^{\calT_\calH}(\F B\lor\F\Btilde) = \Prob_\mu^{\calT_\calH}(\F
		B) = 1 - \Prob_\mu^{\calT_\calH}(\GG\comp{B}) < 1
		\puncteq{,}
	\]
	which means that $\calT_\calH$ is not decisive w.r.t.\ $B$ from $\mu$.
\end{example}

\begin{proposition} \label{prop:strongResetAreDec}
	Every \cyclereset{} SHS is decisive w.r.t.\ any measurable set.
\end{proposition}
\newcommand{\Lsr}{\ensuremath{L^\mathsf{SR}}}
\newcommand{\Esr}{\ensuremath{E^\mathsf{SR}}}
\begin{proof}
	Let $\calH$ be a \cyclereset SHS, and $B\in\Sigma_{\calH}$ be a measurable
	set. We define
	\begin{align*}
		\Lsr &= \{ (\ell,\ell')\in L^2\mid \exists e'=(\ell, a, \ell')\in E
		\land \forall e=(\ell, a, \ell')\in E, e\text{ is strongly reset}\}
		\puncteq{,}\\
		\Esr &= \{ e = (\ell, a, \ell')\in E\mid (\ell, \ell')\in \Lsr\}
		\puncteq{.}
	\end{align*}
	We show that $\calT_\calH$ is strongly decisive w.r.t.\ $B$, using the
	criterion from Proposition~\ref{prop:mainDecCritGen}.
	In this proof, for $\ell$ a location of $\calH$, we abusively write $\ell$ instead of
	$\{\ell\}\times \IR^n$ in \LTL formulae and in conditional distributions.
	Since every infinite run goes infinitely often through at
	least one cycle, and $\calH$ is \cyclereset{}, it holds that for all
	$\nu\in\Dist(S_\calH)$,
	\[
	\Prob_\nu^{\calT_\calH}\left(\bigvee_{(\ell,\ell')\in \Lsr} \GG\F
	(\ell\land \X \ell')\right) = 1
	\puncteq{.}
	\]
	We set
	\[
	p = \min_{e\in \Esr, \resDist^*_e(\Btilde) < 1}
	\Prob^{\calT_\calH}_{(\resDist^*_e)_{\comp{(\Btilde)}}}(\F B) > 0
	\]
	if there is $e\in \Esr$ such that $\resDist^*_e(\Btilde) < 1$, and $p = 1$
	otherwise.
	Let $\nu\in\Dist(S_\calH)$, $(\ell,\ell')\in \Lsr$. We have that
	$\nu_{\ell,\ell'}$ is a linear
	combination of strong reset distributions. If $\nu_{\ell,\ell'}(\Btilde) <
	1$, we then have that
	\[
	\Prob^{\calT_\calH}_{(\nu_{\ell,\ell'})_{\comp{(\Btilde)}}}(\F B) \ge p
	\puncteq{.}
	\]
	We thus satisfy the hypotheses of Proposition~\ref{prop:mainDecCritGen},
	and $\calT_\calH$ is decisive w.r.t.\ $B$. \qed
\end{proof}

Placing (at least) one strong reset per simple cycle is an easy syntactic way
to guarantee that almost surely, infinitely many strong resets are
performed, which is the actual sufficient property used in the proof.  As
there are only finitely many edges, we can find a probability lower
bound $p$ on the probability to reach $B$ after any strong reset, as
required in the criterion of
Proposition~\ref{prop:mainDecCritGen}. Notice that as shown in
Example~\ref{ex:stochPacman}, having independent flows for each
variable and resetting each variable once in each cycle is not
sufficient to obtain decisiveness; variables need to be reset \emph{on
the same discrete transition} in each cycle.

\subsection{Existence of a Finite Abstraction} \label{sec:finiteAbs}
We show that \cyclereset{} SHSs admit a finite $\alpha$-abstraction.
We first give a simple example showing that without one strong reset per cycle,
some simple systems do not admit a finite $\alpha$-abstraction compatible with
the locations.

\begin{example} \label{ex:noFiniteAbs}
	\begin{figure}[tbh]
		\centering
		\includegraphics[width=0.35\columnwidth]{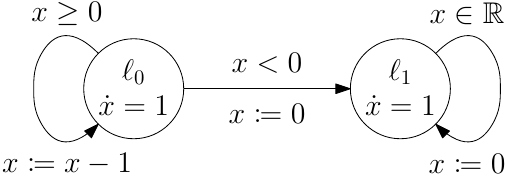}
		\caption{The time delays are given by exponential
			distributions from any state; resets are Dirac distributions.
			The smallest abstraction compatible with $\{\ell_1\}\times\IR$ is
			denumerable.}
		\label{fig:noFiniteAbs}
	\end{figure}
	Consider the SHS of Figure~\ref{fig:noFiniteAbs}. The
	self-loop edge of $\ell_0$ is the only edge not being strongly reset.
	We assume that we want to have an abstraction compatible with
	$s^*\defeq\{\ell_1\} \times\IR$.
	Using Procedure~\ref{alg:abs}, we have to split $\{\ell_0\}\times\IR^n$ in
	$s_0\defeq\Pre^{\calT_\calH}(s^*) = \{\ell_0\}\times\intervaloo{-\infty,0}$
	and $\{\ell_0\}\times\intervalco{0,+\infty}$, as all the states of $s_0$
	can reach $s^*$ with a
	positive probability in one step, but none of the states of
	$\{\ell_0\}\times\intervalco{0,+\infty}$ can. Then,
	$\{\ell_0\}\times\intervalco{0,+\infty}$ must also be split into
	$s_1\defeq\{\ell_0\}\times\intervalco{0,1}$ and
	$\{\ell_0\}\times\intervalco{1,+\infty}$ because the states of $s_1$ can
	all reach $s_0$ with a positive probability in one step, but none of the
	states of $\{\ell_0\}\times\intervalco{1,+\infty}$ can. By iterating this
	argument, we find that the smallest $\alpha$-abstraction compatible with
	$\{\ell_1\}\times\IR$ is denumerable, and the partition that it induces is
	composed of $s^*$, $s_0$ and $s_i = \{\ell_0\}\times\intervalco{i-1, i}$
	for $i\ge 1$. The underlying hybrid system actually belongs to the class of
	\emph{updatable timed automata}~\cite{BDFP04}, and the abstraction almost
	coincides with the \emph{region graph} of the automaton.
\end{example}

The \cyclereset{} assumption is sufficient to guarantee the existence of a
finite abstraction, as formulated in the next proposition.
Its proof consists of showing that Procedure~\ref{alg:abs} terminates under \cyclereset{} assumption. With our
decisiveness result, we can even show that this abstraction is \emph{sound}.

\begin{proposition} \label{prop:finiteAbs}
	Let $\calH$ be an SHS, and $B\in\Sigma_\calH$. If $\calH$ is \cyclereset,
	it has a finite and sound abstraction compatible with $B$ and with the
	locations.
\end{proposition}
\begin{proof}
	For $e = (\ell, a, \ell')\in E$ and $D\subseteq\IR^n$, we write
	$\Pre^{\calT_\calH}_e(\{\ell'\}\times D)$
	for the set of states of $\{\ell\}\times\IR^n$ that have a positive
	probability to reach $D$ through edge $e$.\footnote{Formally,
		the set
		$\Pre^{\calT_\calH}_e(\{\ell'\}\times D)$ can be expressed as
		\[
		\{s = (\ell,\ve)\in \{\ell\}\times\IR^n\mid
		\int_{\tau \in \IR^+} \left(
		\edgDist_{s+\tau}(e) \cdot
		\int_{\ve' \in \IR^n}
		\ind{D}(\ve')\, \ud (\resDist_e(\gamma_\ell(\ve, \tau)))(\ve') \right)
		\ud \delDist_s(\tau)
		> 0\}
		\puncteq{.}
		\]}
	Note that $\Pre^{\calT_\calH}(\{\ell'\}\times D) =
	\bigcup_{e=(\ell,a,\ell')} \Pre^{\calT_\calH}_e(\{\ell'\}\times D)$.
	The fact that we consider transitions combining both a
	continuous evolution and a discrete step is crucial; every transition has
	to go through an edge.

	We show that Procedure~\ref{alg:abs} terminates, with
	\[
	\calP_{init} =
	\bigcup_{\ell\in L} \big\{B\cap (\{\ell\}\times\IR^n),
	\comp{B}\cap(\{\ell\}\times\IR^n)\big\}\setminus\{\emptyset\}\puncteq{.}
	\]
	Now let $e = (\ell, a, \ell')$ be a \emph{strongly reset} edge. Notice that
	for any $D\subseteq\IR^n$,
	\[
	\Pre^{\calT_\calH}_e(\{\ell'\}\times D) =
	\begin{cases}
		\Pre^{\calT_\calH}_e(\{\ell'\}\times \IR^n) &\text{if}\ \resDist_e^*(D)
		> 0\puncteq{,}\\
		\emptyset &\text{if}\ \resDist_e^*(D) = 0\puncteq{.}
	\end{cases}
	\]
	Indeed, states that can take edge $e$ with a positive probability
	will always reach the same parts of the partition of
	$\{\ell'\}\times\IR^n$.
	Hence, edge $e$ induces one refinement in the states of location $\ell$, but
	no matter how the partition of $\{\ell'\}\times\IR^n$ is refined, it will
	not induce an extra refinement of the partition of $\{\ell\}\times\IR^n$,
	thanks to the strong reset.
	This may not be the case without strong reset, as shown in
	Example~\ref{ex:noFiniteAbs}.

	The procedure propagates each extra refinement following the edges
	backwards, but no extra refinement is propagated past strongly reset edges.
	As there is at least one strong reset per cycle, this implies that for each
	refinement, the number of iterations of the procedure is finite. Hence, it
	terminates and produces a finite abstraction.
	As cycle-reset SHSs are decisive w.r.t.\ any measurable set, we immediately
	obtain by Proposition~\ref{prop:decImpliesSound} that this abstraction is
	sound. \qed
\end{proof}

	\section{Reachability Analysis in Cycle-Reset Stochastic Hybrid Systems}
		\label{sec:ominimal}
		Our goal in this section is to perform a qualitative and quantitative analysis
of \cyclereset{} stochastic hybrid systems, using the properties established in
the previous section.
A first hurdle to circumvent is that arbitrary stochastic hybrid systems are
very difficult to apprehend algorithmically: for instance, the continuous
evolution of their variables may be defined by solutions of systems of
differential equations, which we do not know how to solve in general.
To make the problem more accessible, we follow the approach
of~\cite{BMRT04,LPS00} for non-deterministic hybrid systems by assuming that
some key components of our systems are definable in a mathematical structure.
This restricts the syntax of our SHSs in such a way that, as we will
show with a few extra hypotheses, qualitative and quantitative reachability
problems become decidable.

We adapt this point of view to the stochastic framework.
In Section~\ref{sec:qualReachAnal}, we formulate simple
decidability assumptions under which the finite $\alpha$-abstraction from
Section~\ref{sec:finiteAbs} is computable,
which makes qualitative reachability problems decidable. We identify a large
class of SHSs satisfying these hypotheses, namely \emph{\cyclereset{} o-minimal SHSs defined
in a decidable theory}.
In Section~\ref{sec:quantReachDec}, we also identify sufficient hypotheses for
the approximate quantitative problem to be decidable, in the form of a
finite  set of probabilities that have to be approximately computable.

\subsection{Qualitative Reachability Analysis} \label{sec:qualReachAnal}
We assume that the reader is familiar with the following model-theoretic and logical terms: first-order formula, structure, definability of sets, functions
and elements, theory, decidability of a theory. We refer
to~\cite{Ho97} for an introduction to these concepts. In what
follows, by \emph{definable}, we mean definable without parameters.

\begin{definition}[Definable SHS]
	Given a structure $\calM$, an SHS $\calH$
	is said to be \emph{defined in $\calM$} if
	for every location $\ell\in L$, $\flow_\ell$ is a function definable in
	$\calM$ and $\inv(\ell)$ is a set definable in $\calM$, and
	for every edge $e\in E$, the set $\calG(e)$ is definable and there exists a
	first-order formula $\resetFormula_e(\vect{x},\vect{y})$ such that
	$\ve'\in\calR_e(\ve)$ if and only if $\resetFormula_e(\ve,\ve')$ is true.
\end{definition}

Note that we require that the \emph{flow} of the dynamical
system in each location is definable in $\calM$, and not that it is the
solution to a definable system of differential equations.
We fix $\calM = \langle\IR,<,\ldots\rangle$ a structure, and $\calH$
an SHS.  We make three assumptions:
\begin{enumerate}[label=\textbf{(H\arabic*)},leftmargin=*]
	\item $\calH$ is defined in $\calM$, \label{hyp:defined}
	\item the theory of $\calM$ is decidable, \label{hyp:decidable}
	\item for every $\ell\in L$ and definable $D\subseteq\IR^n$,
	$\Pre^{\calT_\calH}(\{\ell\}\times D)$ is definable.
	\label{hyp:definableDelay}
\end{enumerate}

These hypotheses give us an automatic way to handle many questions
about $\calH$. For instance, for $s = (\ell,\ve)\in S_\calH$ and
$e=(\ell,a,\ell')\in E$, the set $I(s,e)$ (and thus $I(s)$) is definable since
$\inv(\ell)$, $\flow_\ell(\cdot,\cdot)$ and $\calG(e)$ are definable
by~\ref{hyp:defined}, and
\[
	I(s,e) = \{\tau \mid \tau \ge 0 \land
	           \forall \tau'(0\le\tau'\le\tau \implies \flow_\ell(s,\tau')\in
	           \inv(\ell)) \land
               \flow_\ell(s,\tau)\in\calG(e)
               \}
    \puncteq{.}
\]

The main issue when lifting such considerations to the stochastic setting is
that the definability of probabilities and measures is not guaranteed.
For instance, the function $x\mapsto
\frac{1}{x}$ is definable in
$\langle\IR,<,+,\cdot,0,1\rangle$, but its integral is not: for $t\ge 1$,
\[
	\int_1^t \frac{1}{x}\, \ud x = \log(t)
	\puncteq{.}
\]
Therefore, using probability measures given by a definable
probability density function is not sufficient to write arbitrary first-order
formulae about actual probabilities, so $\Pre^{\calT_\calH}(\{\ell\}\times D)$
is not necessarily definable, even for $D$ definable.
To compensate, we make assumption~\ref{hyp:definableDelay}, which amounts to
assuming that distributions $\delDist_s$ for $s\in S_\calH$
and $\resDist_e(\ve)$ for $e\in E$ and $\ve\in\IR^n$ still have some properties
that are sufficient to decide whether their evaluation on definable sets is
positive.
We summarize in the following proposition what the previous hypotheses entail.

\begin{restatable}{proposition}{decQualReach} \label{prop:decQualReach}
	We assume that $\calH$ is a \cyclereset SHS which satisfies
	hypotheses~\emph{\ref{hyp:defined}}, \emph{\ref{hyp:decidable}},
	and~\emph{\ref{hyp:definableDelay}}.
	Let $B\in S_\calH$ be a measurable and definable set of states and
	$\mu\in\Dist(S_\calH)$ be an initial distribution such that for every
	location $\ell$ and definable set $D\subseteq\IR^n$, we can decide whether
	$\mu(\{\ell\}\times D) > 0$.
	Then we can decide whether $\Prob_\mu^{\calT_\calH}(\F B) = 1$ and whether
	$\Prob_\mu^{\calT_\calH}(\F B) = 0$.
\end{restatable}

In particular, this implies that we can decide whether a definable state is in $\Btilde$, which is not the case in general, as shown through Proposition~\ref{prop:undecSHS}.
\begin{proof}
	Using that $\calH$ is \cyclereset, we
	know by Proposition~\ref{prop:finiteAbs} that there exists a
	sound $\alpha$-abstraction of $\calH$ compatible with $B$, which is a
	finite STS (that is, a finite Markov chain)
	$\calT^*_\calH = (\calP_\calH,2^{\calP_\calH},\kappa_\calH)$. Partition
	$\calP_\calH$ is definable by~\ref{hyp:definableDelay} and
	Procedure~\ref{alg:abs}, and Procedure~\ref{alg:abs} is an algorithm thanks
	to hypothesis~\ref{hyp:decidable}.

	By hypothesis on the initial distribution $\mu$, we can decide for which
	sets $P\in\calP_\calH$ we have $\mu(P) > 0$. Therefore, we can compute an
	initial distribution $\mu'\in\Dist(\calP_\calH)$ which is qualitatively
	equivalent to $\alpha_\#(\mu)$.
	By the properties of an abstraction, it holds that
	\[
	\Prob_\mu^{\calT_\calH}(\F B) = 0\Longleftrightarrow
	\Prob_{\alpha_\#(\mu)}^{\calT^*_\calH}(\F \alpha(B)) = 0
	\Longleftrightarrow
	\Prob_{\mu'}^{\calT^*_\calH}(\F \alpha(B)) = 0
	\puncteq{,}
	\]
	which can be decided for a finite Markov chain.
	Similarly, by Proposition~\ref{prop:almostSureReach}, as $\calT^*_\calH$ is
	a sound $\alpha$-abstraction and is decisive w.r.t.\ $\alpha(B)$ (since its
	state space is finite), it holds that
	\[
	\Prob_\mu^{\calT_\calH}(\F B) = 1\Longleftrightarrow
	\Prob_{\alpha_\#(\mu)}^{\calT^*_\calH}(\F \alpha(B)) = 1
	\Longleftrightarrow
	\Prob_{\mu'}^{\calT^*_\calH}(\F \alpha(B)) = 1
	\puncteq{,}
	\]
	which can also be decided. \qed
\end{proof}

\paragraph{O-minimal SHSs}
We identify a large subclass of SHSs satisfying hypotheses~\ref{hyp:defined},
\ref{hyp:decidable}, and~\ref{hyp:definableDelay}. This subclass
consists of the \emph{o-minimal SHSs}, i.e., SHSs
defined in an \emph{o-minimal structure} (introduced in~\cite{vdD84,PS86}),
with additional assumptions on the probability distributions.

\begin{definition}
	A totally ordered structure $\calM = \langle M, <, \ldots\rangle$ is
	\emph{o-minimal} if every definable subset of $M$
	is a finite union of points and open intervals (possibly unbounded).
\end{definition}

In other words, the
definable subsets of $M$ are exactly the ones that are definable with
parameters in $\langle M, < \rangle$.  Some well-known structures are
o-minimal: the ordered additive group of rationals
$\langle\IQ,<,+,0\rangle$, the ordered additive group of reals $\Rlin
= \langle\IR,<,+,0,1\rangle$, the ordered field of reals $\Ralg =
\langle\IR,<,+,\cdot,0,1\rangle$, the ordered field of reals with the
exponential function $\Rexp = \langle\IR,<,+,\cdot,0,1,e^x
\rangle$~\cite{Wi96}.

There is no general result about the decidability of the theories of o-minimal
structures. A well-known case is the
Tarski-Seidenberg theorem, which asserts that there exists a
quantifier-elimination algorithm for sentences in the first-order language of
real closed fields~\cite{Tar48}. This result implies the decidability of the theory of $\Ralg$.
However, it is not known whether the theory of $\Rexp$ is decidable. Its decidability is implied by \emph{Schanuel's	conjecture}, a famous unsolved problem in transcendental number theory~\cite{MW}.

The o-minimality of $\calM$ implies that definable subsets of $M^n$ have a
very ``nice'' structure, described notably by the
\emph{cell decomposition theorem}~\cite{KPS86}. This implies in
particular that every subset of $\IR^n$ definable in an o-minimal structure
belongs to the $\sigma$-algebra of Borel sets of
$\IR^n$~\cite[Proposition~1.1]{Kai12}.
We have in addition the following result, showing that in an o-minimal structure, definable sets with positive Lebesgue measure coincide with definable sets with non-empty interior.

\begin{lemma}[{\cite[Remark~2.1]{Kai12}}] \label{lem:lebesgueOpen}
	Let $\lebesgue{n}$ be the Lebesgue measure on $\IR^n$. Let $\calM$ be an
	o-minimal structure. If $A\subseteq\IR^n$ is definable in
	$\calM$, then $\lebesgue{n}(A) > 0$ if and only if $A^\circ \neq \emptyset$,
	where $A^\circ$ denotes the interior of $A$.
\end{lemma}

In this work, we simply define an \emph{o-minimal SHS} as an SHS defined in an o-minimal structure.

\begin{remark}
The definition of  \emph{o-minimal (non-deterministic) hybrid system} in the
literature usually assumes that all edges are strongly reset. \emph{O-minimal
hybrid systems} were first introduced in~\cite{LPS00}, and further studied
notably in~\cite{BMRT04}. The strong reset hypothesis was relaxed
in~\cite{Gen05} to ``one strong reset per cycle''.
The main result showing the interest of such hybrid systems
is that they admit a finite abstraction (called \emph{time-abstract
bisimulation} in this case), which is computable when the underlying theory is
decidable. This finite abstraction does not necessarily extend to a
``stochastic'' abstraction (as defined in Section~\ref{sec:abstractions}), as
there may be transitions that have probability $0$
to happen, and the corresponding sets of states may not be definable without
hypothesis~\ref{hyp:definableDelay}.
\end{remark}

Let $\calM = \langle\IR,<,+,\ldots\rangle$ be an o-minimal structure whose
theory is decidable, such as $\Ralg$. Let
$\calH=(\calH',\delDist_L,\resDist_\calR,\edgDist)$ be a \cyclereset
o-minimal SHS defined in $\calM$. It therefore satisfies
hypotheses~\ref{hyp:defined} and~\ref{hyp:decidable}.
We still lack some hypotheses about the definability of the probability
distributions to obtain the definability of the finite abstraction. Let
$\mu\in\Dist(S_\calH)$ be an initial distribution. We make the following
assumptions,
which we denote by \hyp:
\begin{itemize}
	\item for all $s = (\ell,\ve)\in L\times\IR^n$, if $I(s)$ is finite, then
	$\delDist_s$ is equivalent to the uniform discrete distribution on $I(s)$;
	if $I(s)$ is infinite, then $\delDist_s$ is equivalent to the Lebesgue
	measure on $I(s)$;
	\item for each $\ell\in L$, the conditional initial distribution
	$\mu_{\{\ell\}\times\IR^n}$ is either equivalent to the discrete measure on
	some finite definable support $D_{\ell}$, or equivalent to the Lebesgue
	measure on a definable support $D_{\ell}$;
	\item for $e\in E$, $\ve\in\IR^n$, we require that $\calR_e(\ve)$ is either
	finite or has non-zero Lebesgue measure $\lebesgue{n}$; $\resDist_e(\ve)$
	is respectively either equivalent to the discrete measure on $\calR_e(\ve)$
	or equivalent to the Lebesgue measure on $\calR_e(\ve)$.
\end{itemize}

The first requirement, about the distribution on the time delays, is a standard
assumption in the case of stochastic timed
automata~\cite{BBBBG08,BBBM08,BBB+14}.
Notice that as $I(s)\subseteq\IR^+$ is definable in an o-minimal structure, it
is infinite if and only if it contains an interval with non-empty interior,
hence if and only if $\lebesgue{1}(I(s)) > 0$ (by Lemma~\ref{lem:lebesgueOpen}).
We formulate a similar requirement on the initial distribution and on the reset
distributions, and we restrict their support to be either finite or to have
non-zero Lebesgue measure.
This postulate is quite natural and easily satisfied: for instance, exponential
distributions (resp.\ uniform distributions $\calU(a,b)$) are
equivalent to the Lebesgue measure on $\IR^+$ (resp.\ $\intervalcc{a,b}$).

If a distribution $\nu$ is discrete with finite definable support $T$, for all
definable sets $D$, we can express that $\nu(D) > 0$ as a first-order formula
$(\exists s\in T, s\in D)$.
If $\nu$ is equivalent to the Lebesgue measure on a definable set
$T$ with $\nu(T) = 1$, for all definable sets $D\subseteq \IR^n$,
\begin{align*}
\nu(D) > 0 &\Longleftrightarrow \nu(T\cap D) > 0 \\
&\Longleftrightarrow \lebesgue{n}(T\cap D) > 0 \\
&\Longleftrightarrow (T\cap D)^\circ \neq \emptyset\qquad\quad
\text{(by Lemma~\ref{lem:lebesgueOpen})}\\
&\Longleftrightarrow \exists \vect{x}\in T\cap D,\exists r > 0 \land
(\forall
\vect{y}, \norm{\vect{x} - \vect{y}} < r \implies \vect{y}\in T\cap
D)
\puncteq{,}
\end{align*}
where $\norm{\vect{z}} = \sum_{i=1}^n \abs{z_i}$ is a definable function, as we
have assumed that $+$ is in $\calM$.
The same reasoning can be applied to distinguish whether the supports of the distributions are finite or have positive Lebesgue measure.

\begin{remark}
	We could also consider distributions that are a linear combination of both
	a discrete distribution and a distribution equivalent to the Lebesgue
	measure. This would require for each occurring distribution to have
	distinct first-order formulae to define the finite support and the
	continuous support.
	We choose to omit this generalization in order not to complicate the
	notations.
\end{remark}

Thanks to these hypotheses, we show that we obtain~\ref{hyp:definableDelay}.
Let $e = (\ell, a, \ell') \in E$ be an edge, and $D\subseteq\IR^n$ be a
definable set.
The set of states in $\{\ell\}\times\IR^n$ that can reach $D$ through $e$
\emph{without delay} with a positive probability is given by
\[
	D'=\{\ell\}\times\{\ve'\in\IR^n\mid \ve'\in \calG(e),
	\resDist_e(\ve')(D\cap\calR_e(\ve')) > 0\}
	\puncteq{.}
\]
Therefore, we have that
\begin{align*}
	\Pre^{\calT_\calH}(\{\ell'\}\times D) = \bigcup_{
	(\ell,a,\ell')\in E}\{\ell\}\times
	\{\ve\in\IR^n\mid \delDist_{(\ell,\ve)}(\flow_\ell(\ve,\cdot)^{-1}(D')) >
	0\}
\end{align*}
is definable. By Proposition~\ref{prop:decQualReach}, we conclude that the
qualitative reachability problem is decidable for \cyclereset o-minimal SHSs satisfying \hyp. We summarize these ideas in the next proposition.

\begin{restatable}{proposition}{sohsAreDec}
	Let $\calH$ be a \cyclereset o-minimal SHS defined in a structure
	whose theory is decidable. Let $B\in\Sigma_\calH$ be a definable set and
	$\mu\in\Dist(S_\calH)$ be an initial distribution. We assume that
	assumption \emph{\hyp} holds. Then one can decide whether
	$\Prob^{\calT_\calH}_\mu(\F B) = 1$ and whether $\Prob^{\calT_\calH}_\mu(\F
	B) = 0$.
\end{restatable}
In particular, we can decide the qualitative reachability problems for \cyclereset SHSs defined in $\Ralg$ and satisfying \hyp. Assuming Schanuel's
conjecture~\cite{MW}, we could extend this result to $\Rexp$.

\subsection{Approximate Quantitative Reachability Analysis}
\label{sec:quantReachDec}
In this section, we show under strengthened numerical hypotheses that we
can solve the quantitative reachability problem in \cyclereset{} SHSs.
Let $\calH$ be a \cyclereset{} SHS, $B\in\Sigma_\calH$ and
$\mu\in\Dist(S_\calH)$.  Our goal is to apply the approximation scheme
described in Section~\ref{sec:dec} in order to approximate
$\Prob^{\calT_\calH}_\mu(\F B)$.  To do so, we remind that we require
that $\calT_\calH$ is decisive w.r.t.\ $B$, which is implied by the
\cyclereset{} hypothesis (Proposition~\ref{prop:strongResetAreDec}),
and the ability to compute for all $m\in\IN$, an arbitrarily close
approximation of
\[
\pmYes = \Prob_\mu^\calT(\F[\leq m] B)\ \text{and}\ \pmNo
	= \Prob_\mu^\calT(\comp{B} \U[\leq m]{\Btilde})
	\puncteq{.}
\]
Notice that
\begin{align*}
\pmYes &= \sum_{j=0}^{m}
\Prob_\mu^{\calT_\calH}(\Cyl(\underbrace{\comp{B},\ldots,\comp{B}}_{j\
	\text{times}}, B)) \\
&= \sum_{j=0}^{m} \sum_{(\ell_0,\ldots,\ell_j)\in L^{j+1}}
\Prob^{\calT_\calH}_\mu(
\Cyl(\ell_0\cap\comp{B},\ldots,\ell_{j-1}\cap\comp{B}, \ell_j\cap B))
\puncteq{,}
\end{align*}
where we write $\ell_i$ as a shorthand for $\{\ell_i\}\times\IR^n$.
To compute $\pmYes$, it is thus sufficient to be able to compute, for every
$0\le j\le m$ and for every path $(\ell_0,\ldots,\ell_{j})$ of the graph $(L,
E)$, the probability
\[
	\Prob^{\calT_\calH}_\mu(
	\Cyl(\ell_0\cap\comp{B},\ldots,\ell_{j-1}\cap\comp{B}, \ell_j\cap B))
	\puncteq{.}
\]

Using the \cyclereset{} hypothesis, we show that we can express this
probability as the product of probabilities of paths with bounded length
$b\in\IN$,
where $b$ is the length of the longest path without encountering a
strong reset. If $j$ is greater than $b$, we know that there is necessarily a
smallest index $i \le b$ for which all edges $(\ell_i, a, \ell_{i+1})$ are
strongly reset.
We have
\begin{multline*}
\Prob^{\calT_\calH}_\mu(
\Cyl(\ell_0\cap\comp{B},\ldots,\ell_{j-1}\cap\comp{B}, \ell_j\cap B)) \\
= \Prob^{\calT_\calH}_\mu(
\Cyl(\ell_0\cap\comp{B},\ldots,\ell_i\cap \comp{B}))
 \cdot\Prob^{\calT_\calH}_{\mu_i}(
\Cyl(\ell_{i}\cap\comp{B},\ldots,\ell_{j-1}\cap\comp{B}, \ell_j\cap B))
\puncteq{,}
\end{multline*}
where $\mu_i = \mu_{\ell_0\cap\comp{B},\ldots,\ell_i\cap \comp{B}}$,
using Lemma~\ref{lem:splitCyl}.
As there may be multiple strongly reset edges between $\ell_i$ and
$\ell_{i+1}$, we can rewrite the second factor as
\[
	\sum_{e = (\ell_i,a,\ell_{i+1})}
	\probEdge_e(\mu_i)\cdot
	\Prob^{\calT_{\calH}}_{\resDist^*_{e}}(\Cyl(\ell_{i+1}\cap\comp{B},\ldots,\ell_{j-1}\cap\comp{B},
	\ell_j\cap B))
	\puncteq{,}
\]
where
\[
	\probEdge_e(\mu_i) = \int_{s\in S_\calH} \int_{\tau\in\IR^+}
	\edgDist_{s+\tau}(e)\,\ud\delDist_s(\tau)\,
	\ud\mu_i(s)
\]
is the probability to take edge $e$ in one step from distribution $\mu_i$. We
can then iterate this reasoning for each probability
$\Prob^{\calT_{\calH}}_{\resDist^*_{e}}(\Cyl(\ell_{i+1}\cap\comp{B},\ldots,\ell_{j-1}\cap\comp{B},
\ell_j\cap B))$; notice that this value does not depend on the initial
distribution.
A similar formula can be obtained to compute $\pmNo$, by replacing the final
occurrence of $B$ by $\Btilde$.

\newcommand{\PathsNoSR}{\ensuremath{L^{\mathsf{NoSR}}}}
We write
\[
\PathsNoSR = \{(\ell_0,\ldots,\ell_j) \in L^{j+1}\mid j\in\IN,
\forall 0\le i\le j - 1, \exists e= (\ell_i, a, \ell_{i+1})\in E\ \text{such
that $e$ is non-strongly reset}\}
\]
for the set of paths of the underlying graph of $\calH$ such that strong resets
can be avoided. This set is finite by the \cyclereset{} hypothesis.
To approximate $\Prob^{\calT_\calH}_\mu(\F B)$, it is sufficient to be able to
approximate arbitrarily closely
\begin{itemize}
	\item for all paths $(\ell_0,\ldots,\ell_j)\in\PathsNoSR$, distributions
	$\nu\in\{\mu\}\cup\{\resDist^*_{e'}\mid e'\in E\ \text{strongly reset}\}$,
	the probability
	\[\Prob^{\calT_\calH}_\nu(\Cyl(\ell_0\cap\comp{B},
	\ldots,\ell_{j-1}\cap\comp{B}, \ell_j\cap B))\]
	and the same probability replacing the final occurrence of $B$ by $\Btilde$;
	\item for all edges $e = (\ell,a,\ell')\in E$ strongly reset, paths
	$(\ell_0,\ldots,\ell_{j-1},\ell)\in\PathsNoSR$, and distributions
	$\nu\in\{\mu\}\cup\{\resDist^*_{e'}\mid e'\in E\ \text{strongly reset}\}$,
	the probability
	\[\probEdge_e(\nu_{\ell_0\cap\comp{B}, \ldots,
	\ell_{j-1}\cap\comp{B}, \ell\cap \comp{B}})\puncteq{.}\]
\end{itemize}

Without strong resets (but still decisiveness w.r.t.\ $B$), a similar
scheme may work. However, this would require a way to compute
probabilities involving $\mu_{C_0,\ldots,C_n}$ with arbitrarily long
sequences $(C_i)_{0\le i\le n}$. With strong resets, we can compute
a finite number of probabilities and assemble them to compute $\pmYes$ and
$\pmNo$ for arbitrarily large values of $m$.

\begin{remark}
It is undecidable in general to decide whether some elements of the state space
are in $\Btilde$ (by-product of Proposition~\ref{prop:undecSHS}). However,
the definition of $\Btilde$ is given by a qualitative
reachability property: using the work done in the previous section, and under
the hypotheses of Proposition~\ref{prop:decQualReach}, we can obtain a
first-order formula defining $\Btilde$. These hypotheses can thus help compute
the probabilities involving $\Btilde$ required for the approximate quantitative
problem.
\end{remark}

	\section{Conclusion}
		\subsubsection*{Summary.}
This article presented in Section~\ref{sec:sts} how to solve reachability 
problems in stochastic transition systems (STSs) via the \emph{decisiveness} 
notion, introduced in~\cite{ABM07,BBBC18}. We notably solved in 
Lemma~\ref{lem:GBcGFAGen} a question that was left open in~\cite{BBBC18} about 
the almost-sure reachability of a set of states in the presence of an 
attractor. 
This allowed formulating a general sufficient condition for decisiveness 
in Proposition~\ref{prop:mainDecCritGen}, which encompasses known 
decisiveness criteria from the literature.

From Section~\ref{sec:stochHS} onward, we focused our attention on
\emph{hybrid} models.
We considered a stochastic extension of the classical \emph{hybrid 
systems}, called \emph{stochastic hybrid systems} (SHSs), and gave them a 
semantics as STSs in order to apply the theory developed in 
Section~\ref{sec:sts}.
We showed that the qualitative and 
quantitative reachability problems are undecidable even for reasonably 
well-behaved SHSs.
This result is not surprising, as reachability 
problems are already undecidable for simple classes of non-deterministic 
hybrid systems~\cite{HKPV,HR00}.
We then showed in Section~\ref{sec:decisiveSHS} that SHSs with one \emph{strong 
reset}
per cycle (\emph{\cyclereset{}}) are decisive with respect to any measurable set, using our new decisiveness criterion, and admit a finite abstraction.

We identified in Section~\ref{sec:qualReachAnal} reasonable 
assumptions leading to the effective computability of this abstraction.
These assumptions pertain to the definability of the different components of 
the SHSs (resets, guards, invariants, dynamics, specific properties of 
distributions) in a mathematical structure, and the decidability of first-order 
formulae in this structure.
Combined with the decisiveness results from Section~\ref{sec:decisiveSHS}, the 
finite abstraction can be used to decide qualitative reachability problems. We 
proved that 
\emph{o-minimal SHSs}, which are SHSs defined in an 
\emph{o-minimal} structure, satisfy these hypotheses. When the theory of the 
o-minimal structure is decidable (which is for instance the case of 
$\Ralg = \langle\IR,<,+,\cdot,0,1\rangle$), the nice properties of the definable 
sets allow deciding for a large class of measures if 
definable sets have positive measure. This is sufficient to compute the 
finite abstraction, which can then be used to decide qualitative 
reachability problems.
We ended the article (Section~\ref{sec:quantReachDec}) with sufficient numerical assumptions to solve  
approximate quantitative reachability problems in \cyclereset{} SHSs.

\subsubsection*{Possible extensions and future work.}
We identify some possible extensions of our results.

A first direction of study is to find other classes of decisive
stochastic systems that can be encompassed by our decisiveness
criterion (Proposition~\ref{prop:mainDecCritGen}). In that respect, a
good candidate is the class of \emph{stochastic regenerative Petri
nets}~\cite{HPRV-peva12,PHV16}.  An application of decisiveness
results to \emph{stochastic Petri nets} was briefly discussed
in~\cite[Section~8.3]{BBBC18}, but under severe constraints; we may
be able to relax part of these constraints with the generalized
criterion.

In~\cite[Sections~6 \&~7]{BBBC18}, the authors show that we can reduce the 
verification of a large class of properties (\emph{$\omega$-regular properties})
of STSs to the verification of reachability properties, under decisiveness 
assumptions. This generalization should be transferable to our work with stochastic hybrid systems.

In Section~\ref{sec:ominimal}, we circumvent the issue of the definability of 
measures and their integrals by using a specific property of the o-minimal 
structures (namely, that 
the Lebesgue measure of a definable set is positive if and only if the interior 
of that set is non-empty, which is definable as a first-order formula). 
However, more 
powerful results exist about the compatibility of o-minimal structures and 
measure theory~\cite{Kai12}. Some o-minimal structures are closed under 
integration with respect to a given measure (then called \emph{tame} measure). 
This consideration may help extend our results about o-minimal stochastic
hybrid systems to a larger class that is less restrictive with respect to 
probability distributions. It could also help for the quantitative problem, as 
approximations of probabilities may be definable.

\subsubsection*{Acknowledgments.}
We would like to thank the anonymous reviewers for their valuable advice, which notably helped simplify the proof of Lemma~\ref{lem:GBcGFAGen}.

	\bibliographystyle{elsarticle-num}
	\bibliography{stochomin}
	\appendix
		\section{Expressing Properties of Runs} \label{sec:app-defs}
To express properties of runs of an STS $\calT = (S, \Sigma, \kappa)$, we use a 
notation very similar to 
\LTL for transition systems~\cite{Pnu77}, with each formula characterizing a 
measurable set of runs.
We define a language 
$\calL_{S,\Sigma}$ of formulae. Its syntax is defined by the following grammar:
\[
	\phi ::= B \mid \phi_1\lor\phi_2 \mid \phi_1\land\phi_2 \mid \lnot\phi_1
	         \mid \phi_1 \U[\bowtie\, k] \phi_2 
	\puncteq{,}
\]
where $B\in\Sigma$,
$\phi_1,\phi_2\in\calL_{S,\Sigma}$, 
${\bowtie}\in\{\le,\ge,=\}$,
and $k\in\IN$.
Let $\run = s_0s_1s_2\ldots$ be an run of $\calT$. We denote 
by $\run_{\ge i} = s_is_{i+1}s_{i+2}\ldots$ the run starting at the 
$i$\textsuperscript{th} step of $\run$. For each kind of formula $\phi$, we 
define when $\run$ satisfies formula $\phi$, denoted by $\run\models\phi$:
\begin{alignat*}{3}
	\run &\models B &&\Longleftrightarrow{}
	&&s_0 \in B\puncteq{,} \\
	\run &\models \phi_1\lor\phi_2 &&\Longleftrightarrow{}
	&&\run \models\phi_1\ \text{or}\ \run\models\phi_2\puncteq{,} \\
	\run &\models \phi_1\land\phi_2 &&\Longleftrightarrow{}
	&&\run \models\phi_1\ \text{and}\ \run\models\phi_2\puncteq{,} \\
	\run &\models \lnot\phi_1 &&\Longleftrightarrow{}
	&&\run\not\models\phi_1\puncteq{,} \\
	\run &\models \phi_1\U[\bowtie\, k] \phi_2 &&{}\Longleftrightarrow{}
	&&\exists i\in\IN, i\,\bowtie\,k,\ \text{s.t.}\ \forall 0\le j<i, \run_{\ge 
	j}\models\phi_1\ \text{and}\ \run_{\ge i}\models\phi_2
	\puncteq{.}
\end{alignat*}

For $\phi\in\calL_{S,\Sigma}$ a formula, we write
$
	\ev{\calT}{\phi} = \{\run\in\Runs(\calT)\mid \run\models\phi\}
$
for the set of runs of $\calT$ satisfying $\phi$. A standard result states that 
for any formula $\phi\in\calL_{S,\Sigma}$, the set of runs $\ev{\calT}{\phi}$ 
is measurable.

For an initial distribution $\mu$, we therefore have that 
$\Prob_\mu^\calT(\ev{\calT}{\phi})$
is well-defined. In the article, we 
write $\Prob_\mu^\calT(\phi)$ instead of $\Prob_\mu^\calT(\ev{\calT}{\phi})$ 
for brevity. As is standard, we also denote
$\phi_1\U \phi_2$ for $\phi_1\U[\ge 0] \phi_2$,
$\F[\le n] \phi$ for $S \U[\le n] \phi$, $\F\phi$ for $S\U[\ge 0] \phi$,
$\GG \phi$ for $\lnot\F\lnot\phi$, $\X \phi$ for $S\U[=1]\phi$, and $\X[n]\phi$ 
for $S\U[= n]\phi$.

	\section{Technical Results and Proofs of Section~\ref{sec:mainDecCrit}}
		\label{sec:app-sts}
		As discussed in the main text, we will prove Lemma~\ref{lem:GBcGFAGen} using \emph{L\'evy's zero-one law}.
The current proof owes much to comments by anonymous reviewers---it used to be a much longer \textit{ad hoc} proof that did not rely on martingale theory.
The reader is referred to~\cite[Section~4]{Dur19} for a thorough introduction to conditional expectations, martingales and L\'evy's zero-one law.

Let $\calT = (S, \Sigma, \kappa)$ be an STS.
We define a sequence $(\oalgf_n)_{n\ge 1}$ of $\sigma$-algebras such that for $n \ge 1$,
\[
\oalgf_n = \{\Cyl(A_0, \ldots, A_{n-1})\mid (A_0, \ldots, A_{n-1})\in \Sigma^n\}.
\]
This is a \emph{filtration} (since $\oalgf_0 \subseteq \oalgf_1 \subseteq \oalgf_2 \subseteq \ldots$), and $\sigma(\bigcup_i \oalgf_i)$ is equal to the $\sigma$-algebra $\oalgf$ generated by all the cylinders.

For a given initial distribution $\mu\in\Dist(S)$, we recall that $(S^\omega, \oalgf, \Prob^\calT_\mu)$ is a probabilistic space.
For an $\oalgf$-measurable and bounded real random variable $\var\colon S^\omega \to \IR$, we denote by $\expecSTS{\calT}{\mu}{\var}\in\IR$ the expected value of $\var$ on this probabilistic space.
For $\var\colon S^\omega \to \IR$, $n\in\IN$, we denote by $\expecCondSTS{\calT}{\mu}{\var}{\oalgf_n}\colon S^\omega\to\IR$ the conditional expectation of $\var$ given $\oalgf_n$: this is an $\oalgf_n$-measurable function that maps an infinite run $\rho$ to the expected value of $\var$ given the first $n$ steps of $\rho$.
See~\cite[Section~4.1]{Dur19} for a formal definition of conditional expectation.
We will use two properties of conditional expectations.

\begin{proposition}[L\'evy's zero-one law for STSs] \label{prop:levy}
	Let $\calE \in \oalgf$ be a measurable event, $\mu\in\Dist(S)$ be an initial distribution.
	We have
	\[
	\lim_{n \to \infty} \expecCondSTS{\calT}{\mu}{\ind{\calE}}{\oalgf_n} = \ind{\calE}
	\]
	almost surely.
\end{proposition}

This is a classical result from martingale theory~\cite[Theorem~4.6.9]{Dur19}, applied to infinite runs of STSs with the aforementioned filtration.
Intuitively, it states that given an infinite run $\rho = s_0s_1\ldots$, almost surely, the probability that a run starting from $s_0\ldots s_n$ satisfies $\calE$ converges, as $n\to\infty$, either to $1$ (in which case $\rho\in\calE$) or to $0$ (in which case $\rho\notin \calE$).

Let us now discuss a second property of conditional expectations.
We say that an event $\calE\in\oalgf$ is \emph{shift-invariant} if for all infinite words $\rho\in S^\omega$, finite words $\rho' \in S^*$, $\rho\in\calE$ if and only if $\rho'\rho\in\calE$---in other words, if being an element of $\calE$ is not affected by changing a finite prefix.

\begin{lemma} \label{lem:shiftinvariantCond}
	Let $\calE$ be a shift-invariant event. For any $i \ge 1$, for any initial distribution $\nu\in\Dist(S)$, for any infinite run $\rho = s_0s_1\ldots$,
	\[
	\expecCondSTS{\calT}{\nu}{\ind{\calE}}{\oalgf_i}(\rho) = \Prob^\calT_{\delta_{s_{i-1}}}(\calE).
	\]
\end{lemma}
This equality is very intuitive: if we can observe the first $i$ steps of $\rho$, then the probability to satisfy $\calE$, given that $\calE$ is shift-invariant, is simply the probability to satisfy $\calE$ from $s_{i-1}$.
\begin{proof}
	We need to show that for $i \ge 1$, $\nu\in\Dist(S)$, the $\oalgf_i$-measurable function $f_i\colon \rho = s_0s_1\ldots \mapsto \Prob^\calT_{\delta_{s_{i-1}}}(\calE)$ satisfies the definition of the conditional expectation $\expecCondSTS{\calT}{\nu}{\ind{\calE}}{\oalgf_i}$~\cite[Section~4.1]{Dur19}, that is
	\[
	\forall C\in\oalgf_i, \expecSTS{\calT}{\nu}{\ind{\calE}\ind{C}} = \expecSTS{\calT}{\nu}{f_i\ind{C}}.
	\]
	By the monotone class theorem, it is sufficient to prove this equality for all $C = \Cyl(A_0, \ldots, A_{i-1})$ with $(A_0, \ldots, A_{i-1}) \in \Sigma^{i}$.
	We proceed by induction on $i$.
	Let $i = 1$, $\nu\in\Dist(S)$, $C = \Cyl(A_0)$ with $A_0\in\Sigma$.
	We have
	\begin{align*}
	\expecSTS{\calT}{\nu}{\ind{\calE}\ind{C}} &= \Prob^\calT_\nu(\calE\cap\Cyl(A_0)) \\
	&= \int_{s_0\in A_0} \Prob^\calT_{\delta_{s_0}}(\calE)\ud\nu(s_0) \\
	&= \int_{s_0\in S} \Prob^\calT_{\delta_{s_0}}(\calE)\ind{A_0}(s_0)\ud\nu(s_0) \\
	&= \expecSTS{\calT}{\nu}{f_1\ind{C}}.
	\end{align*}

	Now let $i > 1$, $\nu\in\Dist(S)$, $C = \Cyl(A_0, \ldots, A_{i-1})$ with $(A_0, \ldots, A_{i-1}) \in \Sigma^{i}$.
	We assume that the desired property holds for all $j < i$. We have
	\begin{align*}
	\expecSTS{\calT}{\nu}{\ind{\calE}\ind{C}} &= \Prob^\calT_\nu(\calE\cap\Cyl(A_0,\ldots,A_{i-1})) \\
	&= \int_{s_0\in A_0} \Prob^\calT_{\kappa(s_0, \cdot)}(\calE\cap\Cyl(A_1,\ldots,A_{i-1}))\ud\nu(s_0) &&\text{as $\calE$ is shift-invariant}\\
	&= \int_{s_0\in A_0} \expecSTS{\calT}{\kappa(s_0, \cdot)}{\ind{\calE}\ind{\Cyl(A_1, \ldots, A_{i-1})}}\ud\nu(s_0)\\
	&= \int_{s_0\in A_0} \expecSTS{\calT}{\kappa(s_0, \cdot)}{f_{i-1}\ind{\Cyl(A_1, \ldots, A_{i-1})}}\ud\nu(s_0) &&\text{by induction hypothesis}\\
	&= \int_{s_0\in S} \expecSTS{\calT}{\delta_{s_0}}{f_{i}\ind{\Cyl(A_0, A_1, \ldots, A_{i-1})}}\ud\nu(s_0) &&\text{as $\calE$ is shift-invariant}\\
	&= \expecSTS{\calT}{\nu}{f_i\ind{C}}.
	\end{align*}
	\qed
\end{proof}

We now restate and prove the technical lemma from Section~\ref{sec:mainDecCrit}.

\GBcGFAGen*
\begin{proof}
	In order not to obfuscate the interesting ideas of the proof with technical considerations, we first prove the lemma for $n=0$ (with $\calA = A\in\Sigma$), and explain afterwards how to extend the proof to obtain the general statement.
	We want to prove that for all $\mu\in\Dist(S)$,
	\[
		\Prob_\mu^\calT(\GG \comp{B} \land \GG\F A) = 0\puncteq{.}
	\]
	Let $\mu\in\Dist(S)$ be an initial distribution.
	We assume w.l.o.g.\ that $A\cap B = \emptyset$---indeed, if that is not the case, we simply notice that $\Prob^\calT_\mu(\GG \comp{B} \land \GG\F A) = \Prob^\calT_\mu(\GG \comp{B} \land \GG\F (A\cap \comp{B}))$ and we replace $A$ by $A\cap\comp{B}$ in the rest of the proof.

	Let us consider a modified STS $\calT_B$ which is equal to $\calT$, except that $B$ is made absorbing (we assume that for $s\in B$, $\kappa(s,\cdot)$ is the Dirac distribution $\delta_s$).
	Notice that $\Prob^\calT_\mu(\F B) = \Prob^{\calT_B}_\mu(\F\GG B)$, and $\Prob^\calT_\mu(\GG \comp{B} \land \GG\F A) \le \Prob^{\calT_B}_\mu(\GG\F A)$ (as $A\cap B = \emptyset$, runs that see $A$ infinitely often without seeing $B$ in $\calT$ are just as likely in $\calT_B$).
	Notice also that the event $\F\GG B$ is shift-invariant. We have
	\begin{align*}
	\ev{\calT_B}{&\GG\F A} \\
	&= \{\rho = s_0s_1\ldots\in S^\omega \mid \forall i, \exists j \ge i, s_j \in A\}\\
	&\subseteq \{\rho = s_0s_1\ldots\in S^\omega \mid \forall i, \exists j \ge i, \Prob^{\calT_B}_{\delta_{s_j}}(\F B) \ge p \} &&\text{by hypothesis on $A$} \\
	&= \{\rho = s_0s_1\ldots\in S^\omega \mid \forall i, \exists j \ge i, \Prob^{\calT_B}_{\delta_{s_j}}(\F\GG B) \ge p \} &&\text{by construction of $\calT_B$} \\
	&= \{\rho \in S^\omega \mid \forall i, \exists j \ge i, \expecCondSTS{\calT_B}{\mu}{\ind{\F\GG B}}{\oalgf_{j+1}}(\rho) \ge p \} &&\text{by Lemma~\ref{lem:shiftinvariantCond}, as $\F\GG B$ is shift-invariant} \\
	&\subseteq \{\rho \in S^\omega \mid \lim_{i\to\infty} \expecCondSTS{\calT_B}{\mu}{\ind{\F\GG B}}{\oalgf_i}(\rho) \text{ is not $0$ if it exists}\} \\
	&= \{\rho \in S^\omega \mid \ind{\F\GG B}(\rho) \neq 0 \} &&\text{by Lévy's zero-one law (Prop.~\ref{prop:levy})} \\
	&= \{\rho \in S^\omega \mid \ind{\F\GG B}(\rho) = 1 \} \\
	&= \ev{\calT_B}{\F\GG B}.
	\end{align*}

	All inclusions and equalities are almost sure.
	In $\calT_B$, as $A\cap B = \emptyset$ and $B$ is absorbing, we have that $\Prob^{\calT_B}_\mu(\GG\F A \land \F\GG B) = 0$.
	As $\ev{\calT_B}{\GG\F A} \subseteq \ev{\calT_B}{\F\GG B}$, this implies that $\Prob^{\calT_B}_\mu(\GG\F A) = 0$.

	We conclude
	\begin{align*}
	\Prob^\calT_\mu(\GG \comp{B} \land \GG\F A)
	&\le \Prob^{\calT_B}_\mu(\GG\F A) = 0.
	\end{align*}

	This proof can be adapted for a sequence $\calA = (A_i)_{0\le i\le n}\in\Sigma^{n+1}$.
	To do so, we define a slightly different sequence of $\sigma$-algebras $(\oalgf_i)_{i\ge 1}$ such that
	\[
		\oalgf_i' = \sigma(\{\Cyl(B_0, \ldots, B_{i-1}, A_1,\ldots,A_n)\mid (B_0,\ldots,B_{i-1})\in\Sigma^i\}).
	\]
	For $\rho\in S^\omega$, let $\calA(\rho)$ be the only sequence $A'_0,\ldots,A'_n$ with $A'_i\in \{A_i, \comp{A_i}\}$ such that $\rho\in\ev{\calT}{\phi_{\calA(\rho)}}$.
	Then, we can obtain an equivalent to Lemma~\ref{lem:shiftinvariantCond} for a shift-invariant event $\calE$: for $i \ge 1$, for an initial distribution $\nu\in\Dist(S)$, for an infinite run $\rho = s_0s_1\ldots$,
	\[
	\expecCondSTS{\calT}{\nu}{\ind{\calE}}{\oalgf'_i}(\rho) = \Prob^\calT_{(\delta_{s_{i-1}})_{\calA(\rho_{\ge i-1})}}(\calE),
	\]
	where $\rho_{\ge i-1} = s_{i-1}s_i\ldots$, assuming $(\delta_{s_{i-1}})_{\calA(\rho_{\ge i-1})}$ (a conditional distribution, see Definition~\ref{def:condDist}) is well-defined.
	Every step of the proof can then be modified by replacing $A$ with $\calA$ or $\phi_\calA$, $\oalgf_i$ with $\oalgf'_i$, and the application of Lemma~\ref{lem:shiftinvariantCond} with the previous equality.
	\qed
\end{proof}

Using this lemma, we can prove our new decisiveness criterion.
\decCrit*
\begin{proof}
	We want to prove that $\calT$ is decisive w.r.t.\ $B$, i.e., that for all
	$\mu\in\Dist(S)$, $\Prob_\mu^\calT(\F B\lor \F \Btilde) = 1$.
	Let $\mu\in\Dist(S)$. We show equivalently that
	\[
	\Prob_\mu^\calT(\GG \comp{B}\land \GG \comp{(\Btilde)}) = 0
	\puncteq{.}
	\]
	We write $\calA_j'$ for the sequence $A^{(j)}_0,\ldots,A^{(j)}_{n_j - 1},
	A^{(j)}_{n_j}\cap\comp{(\Btilde)}$. We have that
	\begin{align*}
	\Prob^\calT_\mu(\GG \comp{B}\land\GG \comp{(\Btilde)})
	&=\Prob^\calT_\mu(\GG \comp{B}\land\GG \comp{(\Btilde)} \land
	\bigvee_{0\le j<m}\GG\F \phi_{\calA_j}) \\
	&= \Prob^\calT_\mu(\GG \comp{B}\land\GG \comp{(\Btilde)} \land
	\bigvee_{0\le j<m}\GG\F \phi_{\calA_j'}) \\
	&\le \sum_{0\le j<m} \Prob^\calT_\mu(\GG \comp{B}\land
	\GG\F \phi_{\calA_j'})
	\puncteq{.}
	\end{align*}
	Let $0 \le j < m$, let $\nu\in\Dist(S)$ with $\nu_{\calA_j'}$ well-defined.
	Notice that $\nu_{\calA_j'} = (\nu_{\calA_j})_{\comp{(\Btilde)}}$.
	As it holds that $\Prob^\calT_{\nu_{\calA_j'}}(\F
	B) \ge p$, we obtain by Lemma~\ref{lem:GBcGFAGen} that
	\[
	\Prob^\calT_\mu(\GG \comp{B}\land
	\GG\F \phi_{\calA_j'}) = 0
	\puncteq{.}
	\]
	Hence $\Prob^\calT_\mu(\GG \comp{B}\land\GG \comp{(\Btilde)}) = 0$. \qed
\end{proof}

The next two propositions focus on proving a claim that was stated throughout
the article, which is that the decisiveness criterion from
Proposition~\ref{prop:mainDecCritGen} generalizes those
from~\cite{ABM07,BBBC18}.
\begin{proposition} \label{app:prop:gen1}
	The criterion provided in Proposition~\ref{prop:mainDecCritGen} generalizes
	those found in~\cite[Lemmas~3.4 \&~3.7]{ABM07}.
\end{proposition}
\begin{proof}
	The result of~\cite[Lemma~3.4]{ABM07} states (for countable Markov chains, but we generalize it here to STSs) that if there exists a
	\emph{finite} attractor $A$ for an STS $\calT = (S,\Sigma,\kappa)$, then
	$\calT$ is decisive w.r.t.\ any measurable set $B$.
	For $B\in\Sigma$, using Corollary~\ref{cor:mainDecCrit} and writing $A' = A\cap \comp{(\Btilde)}$, we can thus simply take
	\[
	p = \min_{s\in A'}\Prob^\calT_{\delta_s}(\F B) > 0
	\puncteq{,}
	\]
	or $p = 1$ if $A'$ is empty.

	The result of~\cite[Lemma~3.7]{ABM07} states that \emph{globally coarse}
	Markov chains~\cite[Definition~3.5]{ABM07} are decisive.
	We say that an STS $\calT = (S,\Sigma,\kappa)$ is \emph{globally coarse}
	w.r.t.\ $B\in\Sigma$ if there exists $p>0$ such that for all $s\in S$,
	either
	$\Prob_{\delta_s}^\calT(\F B) \ge p$, or $\Prob_{\delta_s}^\calT(\F B)=0$.
	To prove the decisiveness of a globally coarse STS, we can simply apply
	Corollary~\ref{cor:mainDecCrit} with $A = S$ as the attractor. \qed
\end{proof}

\begin{proposition}
	 The criterion provided in Proposition~\ref{prop:mainDecCritGen}
	 generalizes the one found in~\cite[Proposition~36]{BBBC18}.
\end{proposition}
\begin{proof}
	We state the result of~\cite[Proposition~36]{BBBC18}.
	\begin{claim}
		Let $\calT_2$ be a countable Markov chain such that $\calT_2$ is an
		$\alpha$-abstraction of $\calT_1$.
		\begin{enumerate}
			\item Assume that there is a finite set
			$A_2=\lbrace s_1, \ldots, s_n\rbrace\subseteq S_2$ such that
			$A_2$ is an attractor for $\calT_2$ and
			$A_1= \alpha^{-1}(A_2)$ is an attractor for	$\calT_1$.
			\item Assume moreover that for every
			$\alpha$-closed set $B$ in $\Sigma_1$, there exist $p>0$ and
			$k\in\IN$ such that for every $1\leq i\leq n$:
			\begin{itemize}
				\item for every $\mu\in\Dist(\alpha^{-1}(s_i))$,
				$\Prob^{\calT_1}_{\mu}(\F[\leq k] B)\geq p$, or
				\item for every $\mu\in\Dist(\alpha^{-1}(s_i))$,
				$\Prob^{\calT_1}_{\mu}(\F B)=0$.
			\end{itemize}
		\end{enumerate}
		Then $\calT_1$ is decisive w.r.t.\ every $\alpha$-closed set.
	\end{claim}
	We prove that these hypotheses satisfy our criterion. Let $B$ be an
	$\alpha$-closed set in $\Sigma_1$.
	We use Corollary~\ref{cor:mainDecCrit}, using $A = A_1$ as an attractor
	for $\calT_1$.
	We consider
	\[
	A_2' = \{ s_i \in
	A_2\mid\forall\mu\in\Dist(\alpha^{-1}(s_i)),
	\Prob^{\calT_1}_{\mu}(\F[\leq k] B)\geq p\}
	\puncteq{.}
	\]
	Writing $A'_1 = A_1\cap \comp{(\Btilde)}$, notice that $A_1' =
	\bigcup_{s_i\in A_2'} \alpha^{-1}(s_i)$.
	Hence for all distributions $\mu\in\Dist(A_1')$, we have that
	$\Prob^{\calT_1}_{\mu}(\F B)\geq p$,
	so $\calT_1$ is decisive w.r.t.\ $B$. \qed
\end{proof}

The criterion from~\cite[Proposition~37]{BBBC18} can also be recovered. We omit
to prove it, as this would require to use properties that are established
in the proof of the proposition itself.

	\section{Proofs of Section~\ref{sec:stochHS}}
		\label{sec:app-shs}
		To prove the undecidability of reachability problems for SHSs, we reduce 
the halting problem for two-counter machines to qualitative reachability 
problems for SHSs. We first recall the definition of a two-counter machine.

\begin{definition}[Two-counter machine]
A \emph{two-counter machine} $M$ is a triple
$(\{b_0,\ldots,b_m\}, C, D)$
where $\{b_0,\ldots,b_m\}$ are instructions and $C$, $D$ are two 
counters ranging over the natural numbers. Each instruction $b_i$, $0 
\leq i \leq m$, is of one of the following three types:
\begin{itemize}
	\item an increment (resp.\ decrement) instruction increments (resp.\ 
	decrements) a specific counter by one, and then jumps to an instruction 
	$b_{k_i}$;
	\item a conditional jump instruction branches to one or another 
	instruction based upon whether a specific counter has currently 
	value $0$;
	\item one halting instruction $b_m$ terminates the machine execution.
\end{itemize}
\end{definition}
We assume that before every decrement instruction, a test is done to verify 
that the counter being decremented is not zero; if it is, then its value is 
unchanged. A \emph{configuration} of a two-counter machine is a triple $\exec 
= (i, c, d)$ where $i$ is the number of the current instruction, and 
$c$ and $d$ specify respectively the values of counters $C$ and $D$ before 
executing $b_i$. 
An \emph{execution} is a finite or infinite sequence of configurations 
$\run = \exec_0\exec_1\exec_2\ldots$ such that 
$\exec_0 = (0, 0, 0)$ is the initial configuration, and for $i 
\geq 0$, if $\exec_j = (i, c, d)$, $\exec_{j+1}$ is the new 
configuration after executing instruction $b_i$. We denote by 
$\len(\run) \in \IN \cup \{\infty\}$ the length of  
execution $\run$. If an execution is 
finite, then its last instruction has to be the halting instruction 
$b_m$. 
The problem of deciding whether the unique execution of a two-counter 
machine ends with a halting instruction is undecidable 
\cite{Minsky67}.

\begin{lemma} \label{lem:redShs2cm}
	Let $M = (\{b_0,\ldots,b_m\}, C, D)$ be a two-counter machine. 
	There is an SHS $\calH$, a measurable set $B \in \Sigma_{\calH}$ and
	an initial distribution $\mu \in \Dist(S_\calH)$ such that 
	$\Prob_\mu^{\calT_\calH}(\F B) = 1$ if and only if $M$ halts and 
	$\Prob_\mu^{\calT_\calH}(\F B) = 0$ if and only if $M$ does not 
	halt. 
	Moreover, the probability distributions on time delays of $\calH$ 
	are all purely continuous distributions, guards are defined by 
	linear comparisons of variables and constants, and dynamics are positive 
	integer slopes.
\end{lemma}

\begin{proof}
	We define an SHS $\calH$ such that for each instruction $b_i$ of 
	$M$, there is a location $\ell_i^*$. To encode the value of 
	counter $C$ (resp.\ $D$), we use four variables $x_1, \ldots, x_4$ 
	(resp.\ $y_1, \ldots, y_4$). We will also use four extra 
	variables $z_1,\ldots,z_4$, for a total of twelve continuous 
	variables ($|X| = 12$). We explain how we encode the value of 
	counter $C$ in our SHS $\calH$; the same method is used for counter $D$.
	If $C = n$, we will require that
	\begin{equation} \label{eq:encoding}
	\frac{1}{2^{n + 1}} < x_2 - x_1 < x_4 - x_3 <\frac{1}{2^n}
	\end{equation}
	almost surely in some key locations of $\calH$.
	
	We show how we perform each type of instruction with our 
	encoding.
	We only show the instructions concerning counter $C$.
	Notice that we can always preserve the encoding of the value of a counter 
	by setting the evolution rate of all the variables used for its 
	encoding to 1. 
	Let $0 \leq i \leq m$ be the index of an instruction of $M$.
	\begin{itemize}
		\item If $b_i$ increments $C$:
		the SHS $\calH^{inc}_{C}$ incrementing counter $C$ is shown in 
		Figure~\ref{fig:increment}. Notice that if $x_1^*,\ldots, 
		x_4^*\in\IR$ are the values of $x_1, \ldots, x_4$ when entering 
		$\ell_i^*$ 
		and encode that $C = n$ as in Equation~\eqref{eq:encoding}, then in 
		location $\ell^{**}_i$, almost surely,
		\[
		\frac{1}{2^{n+2}} < z_2 - z_1 < z_4 - z_3 < \frac{1}{2^{n+1}}
		\puncteq{.}
		\]
		Indeed, if we set $d= x_2^*-x_1^*$, then the time spent in the 
		second location of $\calH^{inc}_C$ is at least $\frac{d}{2} > 
		\frac{1}{2^{n+2}}$, as we need to wait until $x_1>x_2$. Hence in 
		$\ell_i^{**}$, $
		\frac{1}{2^{n+2}} < z_2 - z_1$. Clearly, we also have that $z_2 - z_1 < 
		z_4 - z_3$. If we set 
		$d'= x_4^*-x_3^*$, then the sum of the times spent in the two 
		central 
		locations is $z_4-z_3 < \frac{d'}{2} < \frac{1}{2^{n+1}}$ as $x_3$ has 
		not yet become greater than $x_4$.
		\begin{figure}[tbh]
			\centering
			\fbox{\begin{minipage}{0.25\columnwidth}{\footnotesize
						In $\ell^*_i$ and $\ell_{i}^{**}$: \\$\forall v \in X, 
						\dot{v} 
						= 1$. \\
						In the other locations: \\
						$\forall 1 \leq j \leq 4, \dot{z_j} = 1$, \\
						$\dot{x_1}=\dot{x_3}=3,\newline \dot{x_2}=\dot{x_4}=1$.}
			\end{minipage}}
			\hspace{5pt}
			\begin{minipage}{0.67\columnwidth}
				\includegraphics[width=\textwidth]{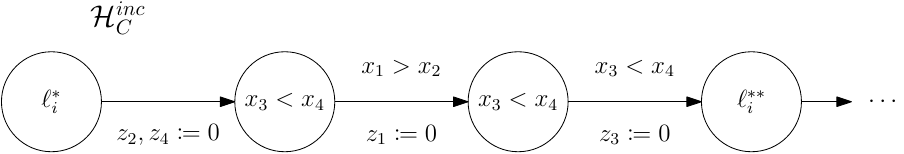}
			\end{minipage}
			\caption{SHS $\calH_C^{inc}$ incrementing the value of $C$. The
				distribution on the time delay in $\ell^*_i$ is any exponential 
				distribution, and in the next two locations it is a uniform 
				distribution on the times after which the outgoing edge is 
				enabled.}
			\label{fig:increment}
		\end{figure}

		A visual representation of the time it takes for $x_1$ to 
		become larger than $x_2$ in the second location is shown in 
		Figure~\ref{fig:slopesXiInc}.
		Therefore the $z_i$ variables satisfy almost surely the requirement for 
		a counter to be equal to $n + 1$ in location $\ell_i^{**}$. To recover 
		the right encoding 
		using the $x_i$ variables, we simply have to append to 
		$\calH^{inc}_{C}$ the exact same SHS where for $1 \leq j \leq 4$, $x_j$ 
		and $z_j$ have been swapped and $\dot{z_1} = \dot{z_3} = 2$ 
		instead of $3$.
		\item If $b_i$ decrements $C$: we use an SHS $\calH_C^{dec}$ which is 
		almost the same SHS as $\calH_C^{inc}$ in 
		Figure~\ref{fig:increment}, but we modify the evolution 
		rates in the two central locations: for $1 \leq j \leq 4$, 
		$\dot{z_j} = 2$, $\dot{x_1}=\dot{x_3}=2$, 
		$\dot{x_2}=\dot{x_4}=1$. A visual representation of the time it takes 
		for $x_1$ to overtake $x_2$ in the second location is shown in 
		Figure~\ref{fig:slopesXiDec}.
		\item If $b_i$ tests whether $C = 0$: the part of the SHS simulating 
		this instruction is shown in 
		Figure~\ref{fig:test}. No matter at what time the transition is taken, 
		only one of the edges is enabled. Notice that $C = 0$ in a location 
		$\ell_i^*$ if and only if $\frac{1}{2} < x_2 - x_1 < 1$ almost surely.
		
		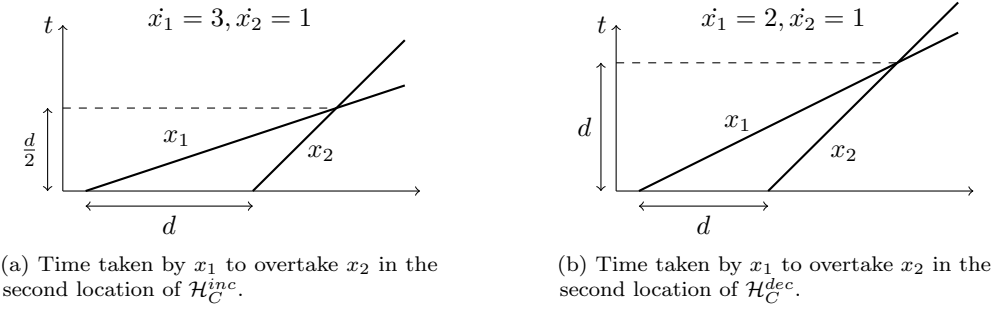
\begin{figure}[htb]
			\centering
			\hspace*{\fill}
			\subfloat[Time taken by $x_1$ to overtake $x_2$ in the second 
			location of $\calH_C^{inc}$.]
				{
				\begin{tikzpicture}
				\draw (2.2,2) node[anchor=south] {$\dot{x_1} = 3,
					\dot{x_2} = 1$};
				\draw[->] (0,0) -- (0,2.2) node[anchor=east] {$t$};
				\draw[<->] (-0.2,0) -- (-0.2,1.1);
				\draw (-0.2,0.55) node[anchor=east] {$\frac{d}{2}$};
				\draw[->] (0,0) -- (4.7,0) node[anchor=north] {};
				\draw[<->] (0.3,-0.2) -- (2.5,-0.2);
				\draw (1.4,-0.2) node[anchor=north] {$d$};
				\draw[dashed] (0,1.1) -- (3.6,1.1);
				\draw[thick] (0.3,0) -- (4.5, 1.4);
				\draw (1.5,0.7) node {$x_1$}; 
				\draw[thick] (2.5,0) -- (4.5,2);
				\draw (3.4,0.5) node {$x_2$}; 
				\end{tikzpicture}
				\label{fig:slopesXiInc}
			}
			\hfill
			\subfloat[Time taken by $x_1$ to overtake $x_2$ in the second 
			location of $\calH_C^{dec}$.]
				{
				\begin{tikzpicture}
				\draw (2.2,2) node[anchor=south] {$\dot{x_1} = 2, 
					\dot{x_2} = 1$};
				\draw[->] (0,0) -- (0,2.2) node[anchor=east] {$t$};
				\draw[<->] (-0.2,0) -- (-0.2,1.7);
				\draw (-0.2,0.85) node[anchor=east] {$d$};
				\draw[->] (0,0) -- (4.7,0) node[anchor=north] {};
				\draw[<->] (0.3,-0.2) -- (2,-0.2);
				\draw (1.15,-0.2) node[anchor=north] {$d$};
				\draw[dashed] (0,1.7) -- (3.7,1.7);
				\draw[thick] (0.3,0) -- (4.5, 2.1);
				\draw (1.6,0.9) node {$x_1$}; 
				\draw[thick] (2,0) -- (4.5,2.5);
				\draw (3,0.5) node {$x_2$}; 
				\end{tikzpicture}
				\label{fig:slopesXiDec}
				}
			\hspace*{\fill}
			\caption{Time taken by $x_1$ to overtake $x_2$ given their 
			evolution rate. The initial difference between $x_2$ and $x_1$ is 
			$d$. Remember that for decrementing, for $1 \leq j \leq 4$, 
			$\dot{z_j} = 2$.}
			\label{fig:slopesXi}
		\end{figure}
	
		\begin{figure}[htb]
			\centering
			\includegraphics[width=0.33\columnwidth]{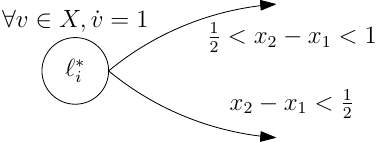}
			\caption{SHS testing whether $C = 0$. The distribution
				on the time delay is any exponential distribution.}
			\label{fig:test}
		\end{figure}
		\item If $b_i = b_m$ is the halting instruction: we model it using
		a single location $\ell_m^*$ with a self-loop edge with no guard or 
		invariant.
	\end{itemize}
	We set $B = \ell_m^* \times \IR^{12}$. The initial distribution 
	$\mu$ assigns $x_1 \defeq 0$, $x_2 \defeq \frac{5}{8}$, $x_3 \defeq 0$, 
	$x_4 \defeq \frac{7}{8}$ in location $\ell_0^*$ ($z_1,\ldots,z_4$ can take
	any value).
	
	We now prove that $M$ halts if and only if 
	$\Prob_\mu^{\calT_\calH}(\F B) = 1$. Let $\run = 
	(i_0, c_0, d_0)(i_1, c_1, d_1)\ldots$ be the execution of $M$.
	For $c, d \in \IN$, let
	\[
	\calV_{c, d} = \{\ve \in \IR^{12} \mid \ve
	\text{ encodes the value $c$ (resp.\ $d$) for counter $C$ (resp.\ $D$)} \}
	\puncteq{.}
	\]
	For an instruction $b_i$ of $M$, we denote by $m_i$ the number of 
	locations used in $\calH$ to represent it. We can prove by 
	induction that for all $0 \leq k < \len(\run)$,
	\begin{equation*}
	\Prob_\mu^{\calT_\calH}(\Cyl(
	\ell^*_{i_0}\times \calV_{c_0, d_0}, 
	\underbrace{S_\calH,\ldots,S_\calH}_{m_{i_0} - 1 \text{ times }},
	\ell^*_{i_1}\times \calV_{c_1, d_1}, 
	\underbrace{S_\calH,\ldots,S_\calH}_{m_{i_1} - 1 \text{ times }},
	\ldots, \ell^*_{i_k}\times \calV_{c_k, d_k}
	)) = 1
	\puncteq{.}
	\end{equation*}
	By definition of $\mu$, we have $\Prob_\mu^{\calT_\calH}(\Cyl(
	\ell^*_{i_0}\times \calV_{c_0, d_0})) = 1$. The induction step is 
	proved by construction of $\calH$; we showed that every 
	instruction could be simulated to preserve almost surely our 
	encoding.
	
	Therefore, if $M$ halts with counter values $c$ and $d$, we get that
	\[
	\Prob_\mu^{\calT_\calH}(\F B) \geq \Prob_\mu^{\calT_\calH}(\F 
	(\ell^*_m\times \calV_{c, d})) = 1
	\puncteq{.}
	\]
	If $M$ does not halt, then location $\ell_m^*$ is almost surely 
	never reached. Thus $\Prob_\mu^{\calT_\calH}(\F B) = 0$. \qed
\end{proof}

As a consequence of this lemma, we can easily prove 
Proposition~\ref{prop:undecSHS}. We recall its statement.

\undecSHSs*
\begin{proof}
	In the previous lemma, we have reduced the 
	halting problem for two-counter machines to the qualitative 
	reachability problem for stochastic hybrid systems. As this halting problem 
	is 
	undecidable~\cite{Minsky67}, using the notations of the previous lemma, we 
	have proved that there is no algorithm to decide if 
	$\Prob_\mu^{\calT_\calH}(\F B) = 1$ or 
	$\Prob_\mu^{\calT_\calH}(\F B) = 0$.
	Moreover, if we could approximate the probability of eventually 
	reaching $B$, since $\Prob_\mu^{\calT_\calH}(\F B)$ is either $0$ 
	or $1$, we could decide if $\Prob_\mu^{\calT_\calH}(\F B)$ 
	is greater or less than $\frac{1}{2}$, and thus decide whether 
	$M$ halts or not. This means that the approximate quantitative 
	reachability problem is also undecidable for any fixed $\epsilon < 
	\frac{1}{2}$. \qed
\end{proof}

\begin{remark}
	Note that the SHS $\calH$ that we have
	built in the proof of Lemma~\ref{lem:redShs2cm} is decisive w.r.t.\ $B$ 
	from $\mu$. We have indeed that either 
	$\Prob_\mu^{\calT_\calH}(\F B) = 1$, or $B$ is almost surely 
	non-reachable from $\mu$ (which means that $\mu(\Btilde) = 1$).
	In both cases, we have $\Prob_\mu^{\calT_\calH}(\F B \lor \F
	\Btilde) = 1$. The key element preventing us from using the procedure 
	to approximate $\Prob_\mu^{\calT_\calH}(\F B)$
	described in Section~\ref{sec:dec} is the computation of $\Btilde$. 
	Deciding whether the initial distribution lies in $\Btilde$ is indeed
	equivalent to deciding whether $M$ halts. This shows that decisiveness, 
	albeit a desirable property to get closer to the decidability frontier, is 
	not sufficient to make any reachability problem decidable, even for very 
	simple (uniform and exponential) distributions, evolution rates, and 
	resets.
\end{remark}

\end{document}